\documentclass[12pt]{article}

\usepackage{amsthm,amssymb,amsmath,amsfonts,appendix,eurosym,geometry,ulem,graphicx,caption,color,setspace,sectsty,comment,footmisc,caption,natbib,pdflscape,subfigure,array,hyperref}

\newtheorem{theorem}{Theorem}
\newtheorem{lemma}{Lemma}
\newtheorem{corollary}{Corollary}
\newtheorem{proposition}{Proposition}
\newtheorem{definition}{Definition}

\newcolumntype{L}[1]{>{\raggedright\let\newline\\arraybackslash\hspace{0pt}}m{#1}}
\newcolumntype{C}[1]{>{\centering\let\newline\\arraybackslash\hspace{0pt}}m{#1}}
\newcolumntype{R}[1]{>{\raggedleft\let\newline\\arraybackslash\hspace{0pt}}m{#1}}

\begin{document}

\begin{titlepage}
\title{Portfolio Optimization Based on Markowitz Stochastic Dominance Criteria}
\author{Peng Xu\thanks{University of Southampton Business School, University of Southampton, Southampton SO16 7QF, England, United Kingdom. All rights reserved. \\ Email: p.xu@soton.ac.uk}}
\date{November 27, 2026}
\maketitle
\begin{abstract}
\noindent This paper develops stochastic optimization models for portfolio selection according to Markowitz Stochastic Dominance (MSD) criteria. Specifically, we first establish dominance conditions in a discrete state-space to capture all inverse S-shaped utility functions consistent with MSD. We then extend these conditions to Weighted Markowitz Stochastic Dominance (MWSD) by accommodating all decision weights generated by inverse S-shaped probability weighting functions in a feasible set. We further demonstrate that the resulting MSD and MWSD conditions can be incorporated as linear constraints into the stochastic optimization models, yielding mixed-integer linear programming (MILP) formulations. Finally, we apply the developed MILP models to financial portfolio optimization and examine the impact of classic behavioral factors, such as reference point and probability distortion, on the portfolio choices of MSD and MWSD investors.\\
\vspace{0in}
\vspace{0in}\\

\bigskip
\end{abstract}
\setcounter{page}{0}
\thispagestyle{empty}
\end{titlepage}
\pagebreak \newpage

\doublespacing

\section{Introduction} \label{sec1}

Stochastic dominance (SD) is a widely used decision rule in finance, economics, and decision sciences. SD allows pairwise comparisons of investment alternatives with uncertain returns without requiring complete information on investors' risk preferences (see, \citealt{Levy2016}, for an overview). In finance, common orders of SD include First-order Stochastic Dominance (FSD; \citealt{QuirkSaposnik1962}), Second-order Stochastic Dominance (SSD; \citealt{HadarRussell1969, HanochLevy1969, RothschildStiglitz1970}), and Third-order Stochastic Dominance (TSD; \citealt{Whitmore1970}). These dominance criteria have found roots in Expected Utility Theory (EUT; \citealt{NeumannMorgenstern1944}), under which investors are assumed to be rational decision makers and to follow the axioms of EUT in portfolio decisions. In particular, FSD assumes that the investor is non-satiable. Any expected utility maximizing investor with a non-decreasing utility function would not prefer a portfolio that is dominated by another according to FSD. Moreover, SSD assumes that the investor is risk-averse. His/her risk preference is captured by a non-decreasing concave utility function. A risk-averse expected utility maximizing investor would not prefer a dominated portfolio in the sense of SSD. In addition to SSD, TSD assumes also that the investor exhibits a preference  for positive skewness in the probability distribution governing the uncertain returns to be ranked. Most SD research has focused on rational decision making under risk and uncertainty, along with the development of SD-constrained stochastic optimization models and their empirical applications (see, among others, \citealt{DentchevaRuszczynski2003, Post2003, Kuosmanen2004, DentchevaRuszczynski2006, KopaPost2009, DentchevaRuszczynski2010, DupacovaKopa2012, DupacovaKopa2014, KopaPost2015, BruniEtAl2017, PostKopa2017, KallioHardoroudi2018, KopaEtAl2018, LiesioEtAl2020, LiesioEtAl2023, PengDelage2024, Xu2024, CesaronePuerto2025, KopaKozmik2025, Neugebauer2026, KopaEtAl2027}). Such SD-based portfolio optimization models typically identify an optimized portfolio that stochastically dominates some pre-specified benchmark return distribution according to the relevant SD criteria.

Over the decades, a large body of theoretical and empirical studies has documented systematic departures from the rational ideals underlying EUT in decision makers' observed behavior, as illustrated by the Allais paradox \cite{Allais1953} and Prospect Theory \cite{KahnemanTversky1979} (see also \citealt{Starmer2000}, for a review). For instance, \cite{MarkowitzSD1952} argues that investors tend to be risk-seeking over gains but risk-averse over losses, where gains and losses are perceived relative to a reference point.  Such risk preferences are often referred to as Markowitz preferences in behavioral finance. \cite{LevyLevy2002} conduct experimental studies using mixed prospects on subjects involving students, university professors, and financial practitioners such as portfolio and mutual fund managers. They find empirical support for characterizing Markowitz preferences with non-linear inverse S-shaped utility functions that are convex in the domain of gains and concave in that of losses. \cite{PostLevy2005} analyze aggregate investor preferences through market efficiency tests based on several SD criteria. In particular, they use the 25 Fama-French and 27 Carhart portfolios formed on market capitalization (size), book-to-market equity (BE/ME) ratio, and momentum to examine whether the market portfolio is SD-efficient. Similarly, their empirical analysis suggests that the market portfolio is efficient for some investors with Markowitz preferences, thereby providing further evidence in support of inverse S-shaped utility functions as well as inverse S-shaped probability weighting functions. Moreover, \cite{EgozcueEtAl2011} study Markowitz preferences and show that investors with inverse S-shaped utility functions allocate their entire wealth to a single asset. However, developing and solving stochastic optimization models involving inverse S-shaped utilities is not straightforward in practice. First, it requires an explicit specification of the investor's risk preferences, but eliciting these preferences can be time-consuming or sometimes impossible (see, e.g., \citealt{Moskowitz1993, LiesioSalo2012, VilkkumaaEtAl2018}). Second, solving a non-linear optimization model can give rise to computational challenges, as any solutions found are not necessarily guaranteed optimal and may require further justification.

\cite{LevyLevy2002} first propose Markowitz Stochastic Dominance (MSD) to characterize risk preferences consistent with the class of inverse S-shaped utility functions through stochastic dominance conditions. Specifically, if one portfolio stochastically dominates another according to MSD, then any investor with a non-decreasing inverse S-shaped utility function would prefer the dominant portfolio. \cite{BaucellsHeukamp2006} extend it to establish Weighted Markowitz Stochastic Dominance (MWSD) by incorporating all inverse S-shaped probability weighting functions in a feasible set. This allows MSD investors to apply decision weights by distorting state probabilities of gains and losses. While these MSD criteria may appear theoretically appealing, they can only help determine whether one portfolio dominates another by MSD or MWSD through pairwise comparison, whereas the utilization of these dominance criteria in stochastic optimization remains unexplored in the operations research and finance literature.

This paper aims to extend the theoretical foundations of Markowitz Stochastic Dominance and advance its methodological frontiers in stochastic optimization. In particular, we first establish the necessary and sufficient dominance conditions based on MSD and MWSD criteria in a discrete state-space. Then, we introduce the first stochastic optimization models based on these MSD criteria allowing incomplete preference information. More specifically, we develop portfolio optimization models utilizing the MSD and MWSD criteria to accommodate incomplete preference information on both utility functions and probability weighting functions (PWFs). Our theoretical and methodological contributions enable us to: (i) capture all MSD preferences, which are typically characterized by non-linear inverse S-shaped utility functions, and (ii) incorporate all decision weights given by a pair of non-linear inverse S-shaped PWFs in a feasible set, using only linear constraints. We further demonstrate that MSD- and MWSD-based portfolio optimization models can be formulated as mixed-integer linear programming (MILP) problems. These MILP models can then identify if there exists an optimal portfolio that is preferred to a given benchmark by all MSD, MWSD investors. Moreover, their application to real-world financial data empirically examines how reference points and probability distortion impact MSD and MWSD investors' portfolio choices.

Moreover, there are several other MSD related research strands worth acknowledging. \cite{ArvanitisTopaloglou2017} develop a consistent non-parametric approach to test the MSD efficiency for dynamic financial time-series data. The test statistics proposed in their work can be computed by solving the corresponding mixed-integer programming (MIP) and linear programming (LP) formulations. They find that the benchmark market portfolio is generally MSD-inefficient to all possible sets of portfolios formed on asset size, book-to-market value, momentum, and industry. \cite{KofinaEtAl2025} recently extend the MSD efficiency test on several global market tracking indices and report consistent international evidence of MSD inefficiency. Moreover, \cite{ArvanitisEtAl2020} develop non-parametric spanning tests (see, e.g., \citealt{ArvanitisEtAl2019}) for MSD and provide their computational formulations based on LP and MIP. Similarly, they also document empirical evidence rejecting the MSD efficiency of the benchmark market portfolio. More broadly, this work is well-positioned within the portfolio decision analysis (PDA; \citealt{SaloEtAl2011}) literature. The models developed here are essentially data-driven ones and can use observations of random variables as direct inputs (\citealt{BertsimasKallus2020}). Consequently, they can be readily applied to behavioral decision analysis in other application contexts, such as project portfolio selection, project scheduling, and resource allocation (for a recent review, see \citealt{LiesioEtAl2021}). In such settings, the decision alternatives (cf. portfolios) to be compared would correspond to feasible solutions of a stochastic optimization problem with uncertainties captured by a finite state-space (\citealt{GustafssonSalo2005, Gutjahr2015, HarrisMazibas2022, CinfrigniniEtAl2025, CesaroneEtAl2026}).

The rest of this paper unfolds as follows. Section \ref{sec2} introduces two Markowitz Stochastic Dominance (MSD) criteria. Section \ref{sec3} establishes dominance conditions for these MSD criteria in a discrete setting. Section \ref{sec4} develops stochastic portfolio optimization models that incorporate MSD and MWSD constraints and presents their resulting mixed-integer linear programming (MILP) formulations. Section \ref{sec5} applies the developed MILP models to empirical financial data and reports results and findings. Finally, Section \ref{sec6} concludes the study and suggests avenues for future research.

\section{Markowitz stochastic dominance criteria} \label{sec2}

An MSD investor exhibits distinct domain-specific preferences to perceived returns, namely risk-seeking over gains but risk-averse over losses, where gains and losses are relative terms defined by a certain reference point, a well-known behavioral factor. Such MSD preferences are typically characterized by some inverse S-shaped utility function $u: [a,b] \rightarrow \mathbb{R}$ that maps individual returns to utilities. The reference point $r$ partitions these returns into gains $[r, b]$ and losses $[a, r]$. We use $U_\mathbb{M}$ to denote the set of all non-decreasing inverse S-shaped utility functions that are convex in gains $t^+ \in [r, b]$ and concave in losses $t^- \in [a, r]$, and $F_X(t)$ is the cumulative distribution function (CDF) of some portfolio $X$.

According to MSD, an investor prefers a portfolio $X$ whose expected Markowitz value
\begin{equation*}
    \mathbb{V}[u(X)] = \int_a^b u(t)dF_X(t) = \int_a^r u(t^-)dF_X(t^-) + \int_r^b u(t^+)dF_X(t^+)
\end{equation*}
is the greatest among any other marketed portfolios under consideration. Formally, any MSD investor would prefer a portfolio whose Markowitz value is greater than or equal to that of another for all inverse S-shaped utility functions in $U_\mathbb{M}$.

\begin{definition}[MSD] \label{def:msd}
Portfolio $X$ weakly dominates portfolio $Y$ by MSD, denoted by $X \succeq Y$, if
$$\mathbb{V}[u(X)] \geq  \mathbb{V}[u(Y)] ~\forall~ u \in U_\mathbb{M}.$$
\end{definition}

\begin{proposition}[\sc \citealt{LevyLevy2002}] \label{prop:msd}
Let $X$ and $Y$ be portfolios. Then $X \succeq Y$ if and only if
\begin{eqnarray}
& &\int_{t^+}^b ~\big[F_Y(t) - F_X(t)\big]dt \geq 0 ~\forall~ t^+ \in [r, b] \label{eq:msd_gains} \\
& \text{ and } \nonumber \\
& &\int_a^{t^-} \big[F_Y(t) - F_X(t)\big]dt \geq 0 ~\forall~ t^- \in [a, r].  \label{eq:msd_losses}
\end{eqnarray}
\end{proposition}

Proposition \ref{prop:msd} formally presents two dominance conditions according to MSD, both of which compare the integrated CDFs of two return distributions. In particular, condition \eqref{eq:msd_gains} states that the integrated CDF of the dominant portfolio $X$ remains below (or equal to) that of the dominated portfolio $Y$ for all return levels in gains $t^+ \in [r,b]$, and condition \eqref{eq:msd_losses} imposes the same requirement for all return levels in losses $t^- \in [a,r]$.

In addition to reference dependence, another widely documented behavioral factor is probability distortion, whereby investors tend to overweight low state probabilities and underweight high state probabilities of the perceived gains and losses. Such distortions are commonly modeled by inverse S-shaped probability weighting functions that are concave for small probabilities and convex for moderate to large probabilities (see, e.g., \citealt{Wakker2010}). Specifically, we use $w^- \in \Omega^{d^-}$ and $w^+ \in \Omega^{d^+}$ to denote the probability weighting functions (PWFs) for losses and gains, respectively, where set $\Omega^d$ contains all increasing inverse S-shaped PWFs $w: [0,1] \rightarrow [0,1]$ that are concave on $[0,d]$ (see Figure 2 of \citealt{BaucellsHeukamp2006}). Note that these PWFs transform cumulative state probabilities rather than individual probabilities of losses and gains into decision weights.

In turn, an MSD investor with an inverse S-shaped utility function, who distorts cumulative state probabilities, prefers a portfolio $X$ whose decision weighted (or distorted expected) Markowitz value
\begin{equation*}
    \mathbb{V}_{w^-}^{w^+}[u(X)] = \int_a^r u(t^-)d\left[w^-\left(F_X(t^-)\right)\right] + \int_r^b u(t^+)d\left[-w^+\left(1-F_X(t^+)\right)\right]
\end{equation*}
is the greatest among all other available portfolios. Therefore, an MSD investor maximizing the decision weighted Markowitz value would prefer a portfolio whose value is no less than that of $Y$ for all inverse S-shaped utility functions in $U_\mathbb{M}$ and for all inverse S-shaped PWFs in sets $\Omega^{d^-}$ and $\Omega^{d^+}$, as formally defined below.

\begin{definition}[MWSD] \label{def:mwsd}
Portfolio $X$ weakly dominates portfolio $Y$ by MWSD, denoted by $X \succeq_{d^-}^{d^+} Y$, if
$$\mathbb{V}_{w^-}^{w^+}[u(X)] \geq \mathbb{V}_{w^-}^{w^+}[u(Y)] ~\forall~ u \in U_\mathbb{M}, ~w^- \in \Omega^{d^-}, ~w^+ \in \Omega^{d^+}.$$
\end{definition}

\begin{proposition}[\sc \citealt{BaucellsHeukamp2006}] \label{prop:mwsd}
Let $X$ and $Y$ be portfolios. Then $X \succeq^{d^+}_{d^-} Y$ if and only if
\begin{eqnarray}
&X \succeq Y \nonumber \\
& \text{ and } \nonumber \\
&F_X(t) \leq F_Y(t)~\forall~t \in [t_d^-, t_d^+), \label{eq:mwsd}
\end{eqnarray}
where
\begin{eqnarray}
t_d^- &= \sup & \bigg( \big\{a\big\} \cup \big\{t\leq r~\big|~F_X(t) \leq d^-,~F_Y(t) \leq d^- \big\} \bigg), \label{eq:tdl} \\ 
t_d^+ &= \inf & \bigg( \big\{b\big\} \cup \big\{t\geq r~\big|~F_X(t) \geq 1-d^+,~F_Y(t) \geq 1-d^+ \big\}\bigg). \label{eq:tdr}
\end{eqnarray}
\end{proposition}

In addition to MSD, Proposition \ref{prop:mwsd} enforces FSD between portfolios $X$ and $Y$ on the interval $[t_d^-, t_d^+)$. Specifically, condition \eqref{eq:mwsd} requires the CDF of the dominant portfolio $X$ to remain below that of the dominated portfolio $Y$ in all return levels within $[t_d^-, t_d^+)$ excluding the upper bound $t_d^+$. The threshold parameters $d^-,d^+ \in [0,1]$ specify the length of this interval. Moreover, increasing the parameter values $d^-$ and $d^+$ decreases the size of the corresponding feasible sets of probability weighting functions, $\Omega^{d^-}$ and $\Omega^{d^+}$, respectively (see Definition \ref{def:mwsd}). As a result, if MWSD holds between portfolios $X$ and $Y$ for certain values of $d^-$ and $d^+$, then the dominance relation holds also for any greater values. Formally, for all values $d^- \leq \tilde{d}^-$ and $d^+ \leq \tilde{d}^+$, it can be implied that
\begin{equation}\label{eq:implied_mwsd_relation}
X \succeq^{d^+}_{d^-} Y \Rightarrow X \succeq^{\tilde{d}^+}_{\tilde{d}^-} Y.
\end{equation}

\section{Establishing Markowitz stochastic dominance conditions in discrete state-space} \label{sec3}

This section develops dominance conditions for the MSD criteria in a discrete state-space. We assume that the underlying state-space $S$ is discrete, consisting of $n$ mutually exclusive and collectively exhaustive states, i.e., $S = \{s_1, \dots, s_n\}$, with the state probabilities $p = (p_1, ..., p_n) \in [0,1]^n$ and $\sum_{i=1}^n p_i =1$. In particular, we consider a pair of portfolios whose returns are modeled by random variables $X$ and $Y$ supported on the interval $[a, b] \subset \mathbb{R}$, which contains all possible return outcomes. The state-specific returns of portfolios $X$ and $Y$ are captured by vectors $x = (x_1,\dots,x_n) = (X(s_1),\dots,X(s_n)) \in \mathbb{R}^n$ and $y = (y_1,\dots,y_n) = (Y(s_1),\dots,Y(s_n)) \in \mathbb{R}^n$, respectively. 

Moreover, the expected return of portfolio $X$ is denoted by 
\begin{equation*}
\mathbb{E}[X]=\sum_{i=1}^n p_i x_i
\end{equation*}
and its cumulative distribution function (CDF) is given by
\begin{equation*} \label{eq:cdf}
F_X(t)=\mathbb{P}\bigg(\bigg\{s_i \in S ~\bigg|~ X(s_i) \leq t \bigg\}\bigg)=\sum_{i \atop x_i \leq t} p_i.
\end{equation*}
Furthermore, integrating the CDF of portfolio $X$ yields
\begin{equation} \label{eq:icdf}
F^{2}_X(t) = 
\int_{- \infty}^t F_X(\tau)~d\tau = \int_{a}^t F_X(\tau)~d\tau = \sum_{i \atop x_i \in [a, t]} p_i (t -x_i) = \sum_{i=1}^n p_i\max\big\{t- x_i, ~0\big\}.
\end{equation}

With this notation, we establish the necessary and sufficient conditions for portfolio $X$ to dominate portfolio $Y$ by MSD. These dominance conditions are formally stated in the following theorem.

\begin{theorem} \label{th:discrete_MSD}
Consider two portfolios $X$ and $Y$. Then $X \succeq Y$ if and only if
\begin{eqnarray}
& F^2_Y(b)-F^2_Y(x_i) \geq F^2_X(b)-F^2_X(x_i) ~ \forall ~ i \in N ~s.t.~ x_i > r \label{eq:th_gains} \\ 
& \text{ and } \nonumber \\
& F^2_Y(y_i) \geq F^2_X(y_i) ~ \forall ~  i \in N ~s.t.~ y_i \leq r. \label{eq:th_losses}
\end{eqnarray}
\end{theorem}
\begin{proof}
See Appendix \ref{appx:proofs}.
\end{proof}

Note that condition \eqref{eq:th_gains} of Theorem \ref{th:discrete_MSD} compares the integrated CDFs of portfolios $X$ and $Y$ in the domain of gains, while its condition \eqref{eq:th_losses} examines them in the domain of losses. Essentially, the integrated CDFs of portfolios $X$ and $Y$ in gains are non-increasing concave piece-wise linear functions, whereas in losses the resulting integrated CDFs are non-decreasing convex piece-wise linear functions (see, e.g., \citealt{LiesioEtAl2020, Xu2024}). Therefore, a key implication of Theorem \ref{th:discrete_MSD} is that, for portfolio $X$ to dominate benchmark $Y$ by MSD, it suffices to check these integrated CDFs only at return levels that correspond to gains of the dominant portfolio $X$ (cf. condition \eqref{eq:th_gains}) and losses of the dominated portfolio $Y$ (cf. condition \eqref{eq:th_losses}). Since the state-space is discrete, these integrated CDFs are evaluated across a finite number of such state-specific portfolio returns.

\begin{theorem} \label{th:discrete_MWSD}
Consider two portfolios $X$ and $Y$. Let $y' \in \mathbb{R}^{n+1}$ be a permutation of the benchmark state-specific returns $y\in \mathbb{R}^n$, augmented by an arbitrary reference point $r$, such that $y'_1 \leq y'_2 \leq \dots \leq y'_n \leq y'_{n+1}$. Then $X \succeq_{d^-}^{d^+} Y$ if and only if
\begin{eqnarray}
& X \succeq Y \nonumber \\
& \text{ and } \nonumber \\
& \tilde{F}_X(y'_i) \leq \max\big\{F_Y(y'_{i-1}), ~d^-\big\}  ~\forall~ i \in \{1,\dots, n+1\}  ~s.t.~ F_Y(y'_{i-1}) < 1-d^+, \label{eq:th_pwfs}   
\end{eqnarray}
where $\tilde{F}_X(t) = \sum_{i|x_i<t}p_i$, $y'_0 \in [a, y'_1)$, and $F_Y(y'_0) = 0$.
\end{theorem}
\begin{proof}
See Appendix \ref{appx:proofs}.
\end{proof}

Theorem \ref{th:discrete_MWSD} formalizes the necessary and sufficient conditions for portfolio $X$ to dominate portfolio $Y$ according to the MWSD criterion. Specifically, besides the MSD conditions outlined in Theorem \ref{th:discrete_MSD}, condition \eqref{eq:th_pwfs} of Theorem \ref{th:discrete_MWSD} states that FSD between portfolios $X$ and $Y$ holds over certain states only. Specifically, establishing FSD between $X$ and $Y$ requires that $\tilde{F}_X$ evaluated at each $y'_i$, i.e., $\tilde{F}_X(y'_i)$ remains below or equal to $F_Y$ evaluated at its preceding state-specific return $y'_{i-1}$, i.e., $F_Y(y'_{i-1})$. Since the state-space is discrete, both $t_d^-$ and $t_d^+$ (see \eqref{eq:tdl} and \eqref{eq:tdr}) would correspond to some state-specific portfolio returns or the reference rate. Thus, condition \eqref{eq:th_pwfs} technically ensures that FSD begins in a state of losses where the maximum of the CDFs of portfolios $X$ and $Y$ evaluated at $t_d^-$ is strictly greater than $d^-$. Subsequently, FSD concludes in a state of gains where the minimum of the CDFs of $X$ and $Y$ evaluated at $t_d^+$ is strictly below $1-d^+$.

Notably, MSD is always a necessary condition for MWSD to be established between portfolios $X$ and $Y$. Hence, MWSD implies MSD but in general not vice versa. Ultimately, these theoretical dominance foundations are essential for the development of portfolio optimization models in the next phase of this work.

\section{Portfolio optimization models based on Markowitz stochastic dominance and weighted Markowitz \\ stochastic dominance} \label{sec4}

This section develops stochastic portfolio optimization models based on the MSD and MWSD conditions introduced in Section \ref{sec3}. Specifically, we first demonstrate that these dominance conditions can be implemented using linear constraints. We then show that MSD- and MWSD-constrained optimization models can be formulated as mixed-integer linear programming (MILP) problems.

We use real-valued random variables $X_1,\dots,X_m$ defined on the discrete state-space $S$ to represent the returns of $m$ outstanding base assets in the asset universe. A portfolio is typically constructed as a convex combination of these base assets. This is technically achieved by introducing a vector of asset weights $\lambda \in \mathbb{R}^m$, which proportionally allocates the initial investment across the $m$ base assets, subject to the no short sales constraint $\lambda \geq 0$. Consequently, the set of all feasible asset weights can be expressed as
\begin{equation*} \label{eq:set_of_weights}
\Lambda=\Bigg\{\lambda \in \mathbb{R}_+^m~\Bigg|~\sum_{j=1}^m \lambda_j=1\Bigg\}.
\end{equation*}

The return of a portfolio weighted by $\lambda$ is captured by random variable $X = \sum_{j=1}^m \lambda_j X_j$ whose state-specific return in state $i$ is given by $x_i=X(s_i)=\sum_{j=1}^m \lambda_j x_{j,i}$. The set of all feasible asset portfolios thus contains all possible return combinations of the base assets, i.e.,
\begin{equation*} \label{eq:set_of_portfolios}
\mathbb{X}=\Bigg\{\sum_{j=1}^m \lambda_j X_j~\Bigg|~\lambda \in \Lambda \Bigg\}.
\end{equation*}
Additionally, let $\mathcal{X} \subset \mathbb{R}^n$ denote the set of all vectors of state-specific returns generated by portfolios in set $\mathbb{X}$. Formally, this set is defined by
\begin{equation*} \label{eq:set_of_returns}
\mathcal{X} = \Big\{ (x_1, ...,x_n) = \big(X(s_1),\dots,X(s_n)\big) \in \mathbb{R}^n ~\big|~ X \in \mathbb{X} \Big\}. 
\end{equation*}

We now formally present MSD- and MWSD-based portfolio optimization models in turn. For an MSD investor with an inverse S-shaped utility function, selecting a benchmark dominating portfolio that maximizes the expected return corresponds to solving the stochastic optimization problem
\begin{equation} \label{eq:concpt_MSD_model}
\max_{X \in \mathbb{X}} \mathbb{E}[X] ~ s.t. ~ X \succeq Y.
\end{equation}

Next, for an MWSD investor whose risk preferences are characterized by an inverse S-shaped utility function and decision weights generated via a pair of inverse S-shaped probability weighting functions (PWFs), identifying an optimal benchmark dominating portfolio amounts to solving the stochastic optimization problem
\begin{equation} \label{eq:concpt_MWSD_model}
\max_{X \in \mathbb{X}} \mathbb{E}[X] ~ s.t. ~ X \succeq_{d^-}^{d^+} Y.
\end{equation}

\subsection{Formulating Markowitz stochastic dominance constraints}

For $X$ to dominate $Y$ by MSD, condition \eqref{eq:th_gains} of Theorem \ref{th:discrete_MSD} states that $F^2_Y(b)-F^2_Y(x_i) \geq F^2_X(b)-F^2_X(x_i)$ holds for all gains such that $x_i > r$. After rearrangement, this condition is rewritten as
\begin{eqnarray}
F^2_Y(b) &\geq& F^2_Y(x_i) + F^2_X(b) - F^2_X(x_i) \overset{\eqref{eq:icdf}}{=} \sum_{y_k \leq x_i} p_k \big(x_i-y_k\big) + \sum_{x_k \leq b} p_k \big(b-x_k\big) - \sum_{x_k \leq x_i} p_k \big(x_i-x_k\big) \nonumber \\
& & = \sum_{k=1}^n p_k \max\big\{x_i -y_k, ~0\big\} +\sum_{k=1}^n p_k \max\big\{b -x_k, ~0\big\} - \sum_{k=1}^n p_k \max\big\{x_i -x_k, ~0\big\} \label{eq:gains1}.
\end{eqnarray}

Clearly, for an arbitrary $x_i > r$, the value of the $k$th term in \eqref{eq:gains1} is jointly determined by the individual output returned by each of its max operators. These outcomes can be evaluated using mixed-integer linear programming (MILP) techniques. In particular, we first introduce non-negative continuous decision variables $\varphi_{i1}, \dots, \varphi_{in}$ and for each $\varphi_{ik}$, we establish the constraint $\varphi_{ik} \geq x_i - y_k$ such that $\varphi_{ik} > 0$ if $x_i > y_k$, otherwise $\varphi_{ik} = 0$, for $k \in N = \{1,\dots,n\}$. Next, we introduce non-negative continuous decision variables $\psi_{1}, \dots, \psi_{n}$. For each $\psi_k$, we formulate the constraint  $\psi_k \geq b - x_k$ such that $\psi_k$ is strictly greater than zero if $b > x_k$, otherwise $\psi_k$ takes a value of zero, for $k \in N$. Then, we construct non-negative decision variables $\theta_{i1},\dots, \theta_{in}$, along with binary decision variables $z_{i1},\dots, z_{in}$ to indicate which of the state-specific returns $x_1,\dots,x_n$ is greater or below $x_i$. For each $\theta_{ik}$, two constraints alternate to bind each state $k$ such that $\theta_{ik} > 0$ and $z_{ik} = 1$ if $x_i > x_k$, or $\theta_{ik} = 0$ and $z_{ik} = 0$ otherwise, for $k \in N$.

Putting all together, for any $x_i > r$, the value of $F^2_Y(x_i) + F^2_X(b) - F^2_X(x_i)$ is equal to the optimal function value of the MILP problem
\begin{eqnarray}
\min_{(\varphi_{i1}, \dots, \varphi_{in}) \in \mathbb{R}_+^n, ~\psi \in \mathbb{R}^n_+, ~ \atop (\theta_{i1},\dots,\theta_{in}) \in \mathbb{R}^{n}_+, ~(z_{i1},\dots,z_{in}) \in \{0,1\}^{n}} & & \sum_{k=1}^{n} p_k \varphi_{ik} + \sum_{k=1}^{n} p_k \psi_k - \sum_{k=1}^{n} p_k \theta_{ik} \label{eq:gains3} \\
s.t. & &\varphi_{ik} \geq x_i - y_k ~\forall~ k \in N  \label{eq:gains4} \\
& & \psi_k \geq b - x_k ~\forall~ k \in N  \label{eq:gains5} \\
& &\theta_{ik} \leq x_i - x_k +M(1-z_{ik}) ~\forall~ k \in N  \label{eq:gains6} \\
& &\theta_{ik} \leq Mz_{ik} ~\forall~ k \in N, \label{eq:gains7}
\end{eqnarray}
where $M$ is a sufficiently large positive constant.

Therefore, establishing condition \eqref{eq:th_gains} of Theorem \ref{th:discrete_MSD} is equivalent to determining if there exists a feasible solution to a system of linear constraints. This result is formally stated by the lemma below.

\begin{lemma} \label{lem:milp_gains}
Consider two portfolios $X$ and $Y$. Then condition $\eqref{eq:th_gains}$ of Theorem \ref{th:discrete_MSD} holds if and only if there exist $\varphi \in \mathbb{R}^{n \times n}_+$, $\psi \in \mathbb{R}^{n}_+$, $\theta \in \mathbb{R}^{n \times n}_+$, and $z \in \{0, 1\}^{n \times n}$ that satisfy the constraints
\begin{eqnarray}
& &\varphi_{ik} \geq x_i - y_k ~\forall~i,k \in N \label{eq:milp_i} \\
& & \psi_k \geq b - x_k ~\forall~ k \in N \label{eq:milp_ii} \\ 
& &\theta_{ik} \leq x_i - x_k + M(1-z_{ik}) ~\forall~i, k \in N \label{eq:milp_iii} \\ 
& &\theta_{ik} \leq Mz_{ik} ~\forall~i,k \in N \label{eq:milp_iv} \\
& &\sum_{k=1}^{n} p_k \varphi_{ik} + \sum_{k=1}^{n} p_k \psi_k - \sum_{k=1}^{n} p_k \theta_{ik} \leq F^2_Y(b) ~\forall~ i \in N  ~s.t.~ x_i > r \label{eq:milp_v},
\end{eqnarray}
where M is a sufficiently large positive constant.
\end{lemma}
\begin{proof}
See Appendix \ref{appx:proofs}.
\end{proof}

Next, for $X$ to dominate $Y$ by MSD, condition \eqref{eq:th_losses} of Theorem \ref{th:discrete_MSD} requires that $F^2_Y(y_i) \geq F^2_X(y_i)$ for all losses such that $y_i \leq r$. In what follows, for an arbitrary $y_i \leq r$, it can be shown that
\begin{equation}
F^2_Y(y_i) \geq F^2_X(y_i) \overset{\eqref{eq:icdf}}{=} \sum_{x_k \leq y_i} p_k \big(y_i-x_k\big) =  \sum_{k=1}^n p_k \max\big\{y_i -x_k, ~0\big\}. \label{eq:losses1}
\end{equation}

Similarly, the value on the right-hand side of \eqref{eq:losses1} can be determined by formulating an LP problem augmented with non-negative continuous decision variables $\delta_{i1}, \dots, \delta_{in}$. Specifically, for each decision variable $\delta_{ik}$, we establish the constraint $\delta_{ik} \geq y_i - x_k$ such that $\delta_{ik} > 0$ if $y_i > x_k$, otherwise $\delta_{ik} = 0$, for $ k \in N$. Therefore, for any $y_i \leq r$, the value of $F_X^2(y_i)$ is equal to the optimal objective function value of the linear programming (LP) problem
\begin{equation}
\min_{(\delta_{i1}, \dots, \delta_{in}) \in \mathbb{R}_+^n}  \left\{\sum_{k=1}^{n} p_k \delta_{ik} ~ \Bigg| ~ \delta_{ik} \geq y_i - x_k ~\forall ~k \in N\right\}. \label{eq:losses2}
\end{equation}

Formally, establishing condition \eqref{eq:th_losses} of Theorem \ref{th:discrete_MSD} thus corresponds to identifying a feasible solution to a system of linear inequalities as presented in the following lemma.

\begin{lemma} \label{lem:milp_losses}
Consider two portfolios $X$ and $Y$. Then condition $\eqref{eq:th_losses}$ of Theorem \ref{th:discrete_MSD} holds if and only if there exists $\delta \in \mathbb{R}^{n^-\times n}_+$ that satisfies the constraints
\begin{eqnarray}
& &\delta_{ik} \geq y_i-x_k ~\forall~ i \in N^-, ~k \in N \label{eq:milp_vi} \\
& &\sum_{k=1}^{n} p_k \delta_{ik} \leq F_Y^2(y_i) ~\forall~ i  \in N^- \label{eq:milp_vii},
\end{eqnarray}
where $N^- = \{ i \in  N ~|~ y_i \leq r\}$ and $n^- = |N^-|$.
\end{lemma}
\begin{proof}
See Appendix \ref{appx:proofs}.
\end{proof}

Lemmas \ref{lem:milp_gains} and \ref{lem:milp_losses} collectively establish the equivalence between the required dominance conditions by MSD and a linear system of constraints in the state-specific returns of portfolios $X$ and $Y$. However, unlike each known $y_i$, it is not possible to determine a priori which state-specific return of portfolio $X$ is greater than the reference point $r$. In addition to this, constraint \eqref{eq:milp_v} in Lemma \ref{lem:milp_gains} applies only to gains where $x_i > r$. Hence, each $x_i$ here can be considered as a decision variable. For this purpose, we introduce binary decision variables $\xi_1, \dots, \xi_n$ and formulate the constraint $x_i - r \leq M \xi_i$ such that $\xi_i = 1$ if $x_i > r$, or $\xi_i = 0$ otherwise, for $i \in N$. Subsequently, this also necessitates a corresponding modification to the left-hand side of constraint \eqref{eq:milp_v} by introducing an extra big $M$-term, as shown below
\begin{equation*}
\sum_{k=1}^{n} p_k \varphi_{ik} + \sum_{k=1}^{n} p_k \psi_k - \sum_{k=1}^{n} p_k \theta_{ik} -M(1-\xi_i),
\end{equation*}
where $M$ is a sufficiently large positive constant. 

With this modification, the theorem below formalizes the MSD conditions in the form of a system of linear constraints.

\begin{theorem} \label{th:milp_msd}
Consider a portfolio $X \in \mathbb{X}$ and a given benchmark $Y$. Then
\begin{enumerate}
    \item[$(i)$] If portfolio $X \in \mathbb{X}$ dominates the benchmark such that $X \succeq Y$, then there exist $\varphi \in \mathbb{R}^{n \times n}_+$, $\psi \in \mathbb{R}^{n}_+$, $\theta \in \mathbb{R}^{n \times n}_+$, $z \in \{0, 1\}^{n \times n}$, $\xi \in \{0, 1\}^{n}$, and $\delta \in \mathbb{R}^{n^- \times n}_+$ that satisfy the system of linear constraints $\eqref{eq:milp_i}-\eqref{eq:milp_iv}$, $\eqref{eq:milp_vi}-\eqref{eq:milp_vii}$, and 
\begin{eqnarray}
& & x_i - r - M\xi_i \leq 0 ~\forall~ i \in N \label{eq:milp_viii} \\
& & \sum_{k=1}^{n} p_k \varphi_{ik} + \sum_{k=1}^{n} p_k \psi_k - \sum_{k=1}^{n} p_k \theta_{ik} -M(1-\xi_i) \leq F^2_Y(b) ~\forall~ i \in N \label{eq:milp_ix},
\end{eqnarray}
where M is a sufficiently large positive constant.
    \item[$(ii)$] If there exist $x = (x_1, \dots, x_n) \in \mathcal{X}$, $\varphi \in \mathbb{R}^{n \times n}_+$, $\psi \in \mathbb{R}^{n}_+$, $\theta \in \mathbb{R}^{n \times n}_+$, $z \in \{0, 1\}^{n \times n}$, $\xi \in \{0, 1\}^{n}$, and $\delta \in \mathbb{R}^{n^- \times n}_+$ that satisfy the system of linear constraints $\eqref{eq:milp_i} - \eqref{eq:milp_iv}$ and $\eqref{eq:milp_vi} - \eqref{eq:milp_ix}$, then portfolio $X \in \mathbb{X}$ is a benchmark dominating portfolio such that $X \succeq Y$.
\end{enumerate}
\end{theorem}
\begin{proof}
See Appendix \ref{appx:proofs}.
\end{proof}

Technically, there can be multiple feasible solutions that explicitly satisfy the set of linear constraints above in Theorem \ref{th:milp_msd}. In fact, each of these solutions can be intuitively interpreted as a portfolio that dominates the benchmark according to the MSD criterion. More specifically, each of these benchmark dominating portfolios $X$ would be preferred to $Y$ by any MSD investor with an inverse S-shaped utility function. Furthermore, the most preferred portfolio can, for instance, be identified as the one that maximizes the expected portfolio return. Consequently, by evaluating the objective function with $\mathbb{E}[X]$, we obtain the following MILP problem
\begin{eqnarray}
\textbf{M1}: \max_{\lambda, ~\varphi, ~\psi, ~\theta \atop z, ~\xi, ~\delta} & & \sum_{i=1}^n p_i \sum_{j=1}^m \lambda_{j}x_{ji} \nonumber \\
s.t.
& & \sum_{j=1}^{m} \lambda_jx_{ji}-\varphi_{ik} \leq y_k ~\forall~ i,k \in N \label{eqn:MSD_gains_i} \\
& & \sum_{j=1}^{m} \lambda_jx_{jk} + \psi_k \geq b ~\forall~ k \in N  \label{eqn:MSD_gains_ii} \\
& & \sum_{j=1}^{m}\lambda_jx_{ji}-\sum_{j=1}^{m}\lambda_jx_{jk}+M(1-z_{ik})-\theta_{ik} \geq 0 ~\forall~ i, k \in N    \label{eqn:MSD_gains_iii} \\
& & Mz_{ik} -\theta_{ik} \geq 0 ~\forall~ i, k \in N   \label{eqn:MSD_gains_iv} \\
& & \sum_{j=1}^{m} \lambda_jx_{ji} - r - M\xi_i \leq 0 ~\forall~ i \in N  \label{eqn:MSD_gains_v} \\
& &\sum_{k=1}^{n} p_k \varphi_{ik} + \sum_{k=1}^{n} p_k \psi_k - \sum_{k=1}^{n} p_k \theta_{ik} -M(1-\xi_i) \leq F^2_Y(b) ~\forall~ i \in N   \label{eqn:MSD_gains_vi} \\
& & \sum_{j=1}^{m} \lambda_jx_{jk} + \delta_{ik} \geq y_i  ~\forall~ i \in N^-, ~ k \in N \label{eqn:MSD_losses_i} \\
& &\sum_{k=1}^{n} p_k \delta_{ik} \leq F_Y^2(y_i) ~\forall~ i  \in N^-   \label{eqn:MSD_losses_ii} \\
& & \lambda \in \Lambda \label{eqn:MSD_devar_i} \\
& & \varphi_{ik}, \theta_{ik} \geq 0 ~\forall~ i,k \in N, \quad \psi_{k} \geq 0 ~\forall~ k \in N, \quad \delta_{ik} \geq 0 ~\forall~ i \in N^-, ~ k \in N \label{eqn:MSD_devar_ii} \\
& & z_{ik} \in \{0,1\} ~\forall~ i, k \in N, \quad \xi_i \in \{0,1\} ~\forall~ i \in N \label{eqn:MSD_devar_iii},
\end{eqnarray}
where $M$ is a sufficiently large positive constant and $N^- = \{ i \in  N ~|~ y_i \leq r\}$.

\begin{corollary} \label{cor:milp_msd}
$X^* = \sum_{j=1}^m \lambda_j^*X_j$ is an optimal solution to the stochastic optimization problem \eqref{eq:concpt_MSD_model} if and only if the MILP problem \textbf{M1} has an optimal solution $(\lambda^*, \varphi^{*}, \psi^{*}, \theta^{*}, z^{*}, \xi^{*}, \delta^{*})$.
\end{corollary}
\begin{proof}
Corollary \ref{cor:milp_msd} follows directly from Theorem \ref{th:milp_msd}.
\end{proof}

\subsection{Formulating weighted Markowitz stochastic dominance constraints}

In order to incorporate MWSD constraints, we first assume for the remainder of this section that the states $\{s_1, \dots, s_n\}$ in $S$ are indexed in ascending order of the state-specific returns of the benchmark portfolio $Y$ such that $y_1 \leq y_2 \leq \dots \leq y_{n-1} \leq y_n$. In addition, we assume that the reference point $r$ coincides with one of the benchmark's state-specific returns, namely $r \in \{y_i~|~ i \in N\}$, where $N = \{1, \dots, n\}$. This assumption entails no loss of generality, as the underlying state-space $S$ can be augmented with a zero-probability state whose state-specific return is set to $r$ without changing the benchmark return distribution.

Under this assumption, we obtain $y'=y$. Theorem \ref{th:discrete_MWSD} states then in condition \eqref{eq:th_pwfs} that $\tilde{F}_X(y_i) \leq \max\big\{F_Y(y_{i-1}), ~d^-\big\}$ across all states in which $F_Y(y_{i-1}) < 1-d^+$. For an arbitrary $y_i$, the value of $\tilde{F}_X(y_i)$ can be modeled as the sum of state probabilities over states $k \in N$ where $x_k$ is strictly less than $y_i$. To achieve this, we construct binary decision variables $\zeta_{i1},\dots,\zeta_{in}$ where each $\zeta_{ik}$ indicates whether $x_k < y_i$. We formulate the constraint $x_k + M\zeta_{ik} \geq y_i$, where $M$ is a sufficiently large positive constant. This constraint ensures $\zeta_{ik}$ = 1 if $x_k < y_i$, otherwise $\zeta_{ik} = 0$. Thus, for any $y_i$, the value of $\tilde{F}_X(y_i)$ is equal to the optimal objective function value of the MILP problem
\begin{equation}
\min_{\zeta_{i1},\dots,\zeta_{in} \in \{0, 1\}^n}  \left\{\sum_{k=1}^n p_k\zeta_{ik} ~\Bigg|~ x_k + M\zeta_{ik} \geq y_i ~\forall~ k \in N\right\}. \label{eq:milp_fsd0}
\end{equation}

Hence, establishing condition \eqref{eq:th_pwfs} of Theorem \ref{th:discrete_MWSD} amounts to determining if there exists a feasible solution to a system of linear inequalities as formally stated by the lemma below.

\begin{lemma} \label{lem:milp_pwfs}
Consider two portfolios $X$ and $Y$. Then condition \eqref{eq:th_pwfs} of Theorem \ref{th:discrete_MWSD} holds if and only if there exists $\zeta \in \{0, 1\}^{n \times n}$ that satisfies the constraints
\begin{eqnarray}
& & x_k + M\zeta_{ik} \geq y_i ~\forall~ i, k \in N \label{eq:milp_fsd1} \\
& & \sum_{k=1}^n p_k\zeta_{ik} \leq \max\big\{F_Y(y_{i-1}), ~d^-\big\}  ~\forall~ i \in N  ~s.t.~ F_Y(y_{i-1}) < 1-d^+, \label{eq:milp_fsd2}
\end{eqnarray}
where $y_0 \in [a, y_1)$, $F_Y(y_0) = 0$, and M is a sufficiently large positive constant.
\end{lemma}
\begin{proof}
See Appendix \ref{appx:proofs}.
\end{proof}

Together, Theorem \ref{th:milp_msd} and Lemma \ref{lem:milp_pwfs} establish the required dominance conditions according to MWSD as a system of linear constraints. This result is formally presented in the following theorem.

\begin{theorem} \label{th:milp_mwsd}
Consider a portfolio $X \in \mathbb{X}$ and a given benchmark $Y$. Then
\begin{enumerate}
    \item[$(i)$] If portfolio $X \in \mathbb{X}$ dominates the benchmark such that $X \succeq_{d^-}^{d^+} Y$, then there exist $\varphi \in \mathbb{R}^{n \times n}_+$, $\psi \in \mathbb{R}^{n}_+$, $\theta \in \mathbb{R}^{n \times n}_+$, $z \in \{0, 1\}^{n \times n}$, $\xi \in \{0, 1\}^{n}$, $\delta \in \mathbb{R}^{n^- \times n}_+$, and $\zeta \in \{0, 1\}^{n \times n}$ that satisfy the system of linear constraints $\eqref{eq:milp_i} - \eqref{eq:milp_iv}$, $\eqref{eq:milp_vi} -\eqref{eq:milp_ix}$, and $\eqref{eq:milp_fsd1} - \eqref{eq:milp_fsd2}$.
    \item[$(ii)$] If there exist $x = (x_1, \dots, x_n) \in \mathcal{X}$, $\varphi \in \mathbb{R}^{n \times n}_+$, $\psi \in \mathbb{R}^{n}_+$, $\theta \in \mathbb{R}^{n \times n}_+$, $z \in \{0, 1\}^{n \times n}$, $\xi \in \{0, 1\}^{n}$, $\delta \in \mathbb{R}^{n^- \times n}_+$, and $\zeta \in \{0, 1\}^{n \times n}$ that satisfy the system of linear constraints $\eqref{eq:milp_i} - \eqref{eq:milp_iv}$, $\eqref{eq:milp_vi}-\eqref{eq:milp_ix}$, and $\eqref{eq:milp_fsd1} - \eqref{eq:milp_fsd2}$, then portfolio $X \in \mathbb{X}$ is a benchmark dominating portfolio such that $X \succeq_{d^-}^{d^+} Y$.
\end{enumerate}
\end{theorem}
\begin{proof}
See Appendix \ref{appx:proofs}.
\end{proof}

Theorem \ref{th:milp_mwsd} can be used to obtain a MILP formulation to identify a benchmark dominating portfolio $X \in \mathbb{X}$ that maximizes the expected return for any MWSD investor with an inverse S-shaped utility function and a pair of inverse S-shaped probability weighting functions producing decision weights.
\begin{eqnarray*}
\textbf{M2}: \max_{\lambda, ~\varphi, ~\psi, ~\theta \atop z, ~\xi, ~\delta, ~\zeta} & & \sum_{i=1}^n p_i \sum_{j=1}^m \lambda_{j}x_{ji} \\
s.t.
& & \eqref{eqn:MSD_gains_i} - \eqref{eqn:MSD_devar_iii} \\
& & \sum_{j=1}^{m}\lambda_jx_{jk} + M\zeta_{ik} \geq y_i ~\forall~ i, k \in N \\
& & \sum_{k=1}^n p_k\zeta_{ik} \leq \max\big\{F_Y(y_{i-1}), ~d^-\big\}  ~\forall~ i \in N  ~s.t.~ F_Y(y_{i-1}) < 1-d^+ \\
& & \zeta_{ik} \in  \{0, 1\} ~\forall~ i, k \in N,
\end{eqnarray*}
where $y_0 \in [a, y_1)$, $F_Y(y_0) = 0$, and $M$ is a sufficiently large positive constant.

\begin{corollary} \label{cor:milp_mwsd}
$X^* = \sum_{j=1}^m \lambda_j^*X_j$ is an optimal solution to the stochastic optimization problem \eqref{eq:concpt_MWSD_model} if and only if the MILP problem \textbf{M2} has an optimal solution $(\lambda^*, \varphi^{*}, \psi^{*}, \theta^{*}, z^{*}, \xi^{*}, \delta^{*}, \zeta^{*})$.
\end{corollary}
\begin{proof}
Corollary \ref{cor:milp_mwsd} follows directly from Theorem \ref{th:milp_mwsd}.
\end{proof}

Note that Corollaries \ref{cor:milp_msd} and \ref{cor:milp_mwsd} demonstrate how Theorems \ref{th:milp_msd} and \ref{th:milp_mwsd} can be used to formulate benchmark dominating expected return maximizers under MSD and MWSD constraints, respectively. In practice, however, the optimal portfolio $X^*$ need not be unique, as multiple portfolios may satisfy the respective dominance constraints and attain the same optimal objective function value. On the other hand, for a given benchmark distribution, the feasible set of the MILP problem \textbf{M1} or \textbf{M2} may be empty. In such cases, there exists no feasible solution (cf. portfolio) that satisfies the MSD or MWSD constraints, implying that the benchmark distribution cannot be dominated according to the corresponding criterion and is therefore justified as an optimal investment alternative for at least some MSD or MWSD preferences.

\section{Application to Industry Portfolio Analysis} \label{sec5}

In this section, we apply the developed stochastic optimization models to empirical financial data on asset returns of industry portfolios. In particular, we assess the efficiency of the benchmark market portfolio for behavioral investors exhibiting MSD and MWSD preferences. More specifically, we test if a benchmark dominating portfolio can be constructed from a convex combination of these base assets such that it would be preferred by all MSD investors with an inverse S-shaped utility function. We then investigate whether such a dominant portfolio can also be identified, which would be preferred by all MWSD investors characterized by an inverse S-shaped utility function and inverse S-shaped probability weighting functions (PWFs). In these tests, we also deploy different reference points and sets of feasible PWFs to characterize and examine MSD and MWSD investors' portfolio choices.

\subsection{Data}

We use monthly return observations of the 49 value-weighted industry portfolios from the Fama-French data library to represent the base assets, the returns of which are modeled by random variables $X_{1},\dots, X_{49}$. Moreover, we utilize monthly return observations of the all-share index from the Center for Research in Security Prices (CRSP) to proxy the benchmark market portfolio $Y$, which tracks the value-weighted returns of all CRSP firms incorporated in the US and listed on the NYSE, AMEX, or NASDAQ exchanges. Specifically, our dataset includes all monthly return observations from Jan/1964 -- Dec/2024, covering a sample period of 61 years, or 732 trading months. This sample choice (see also \citealt{Kuosmanen2004, ArvanitisEtAl2020}) is motivated by the empirical fact that several industries were not listed prior to mid-1963, leading to substantial industry-specific return observations missing from earlier years.

\subsection{Empirical Test Specification}

Our empirical analysis employs a rolling estimation window of $n = 36$ months, which shifts ahead by a 12-month period at a time. Using this estimation approach, our dataset yields a total of 59 overlapping 3-year periods Jan/1964 -- Dec/1966, Jan/1965 -- Dec/1967, \dots, Jan/2022 -- Dec/2024. For each 3-year estimation period, we solve two MILP problems \textbf{M1} (MSD), and \textbf{M2} (MWSD) with two threshold parameter values $d^-=d^+=0.18$ (see Section \ref{sec4}). As part of our sensitivity analysis, preliminary tests were conducted to help determine the appropriate values for these parameters in the MILP problem \textbf{M2} (MWSD). Specifically, it follows from \eqref{eq:implied_mwsd_relation} that if $X \succeq_{0.18}^{0.18} Y$, then the established dominance also holds for any values greater than $d^-=d^+=0.18$, whereas for values sufficiently below $d^-=d^+=0.18$, only FSD can be found in most test cases. Moreover, for each MILP problem, we test particularly two reference points: (i) the risk-free rate $(r=rf)$, and (ii) the median of benchmark returns $(r=md)$, i.e., $r=\text{median}(Y)$. The underlying state-space $S$ of each MILP problem is built using the corresponding return observations of the base assets $X_{1}, \dots, X_{49}$, with a vector of equal state probabilities $p=(\frac{1}{36}, \dots, \frac{1}{36})$. As is common in the SD-based portfolio optimization literature, the general probabilistic assumption here is that each state $s_i \in S$ is equally likely to be realized and historical vectors of the base asset returns can represent plausible scenarios of future return outcomes (\citealt{PostEtAl2018}; see also \citealt{Xu2024}). If the reference point $r$ is not directly observable, e.g., the risk-free rate does not coincide with any of the state-specific returns of the benchmark market portfolio $Y$, then the state-space $S$ is augmented by a zero-probability state whose state-specific return is set to the geometric mean of the monthly T-bill rates (see also \citealt{XuLiesio2026}). The MILP problems \textbf{M1} (MSD) and \textbf{M2} (MWSD) were implemented in Python and solved by Gurobi 11.0.3. The computations were undertaken on a university-managed High Performance Computing (HPC) cluster facility.

\subsection{Results and Findings}

The upper panels in Figure \ref{fig1} illustrate the excess returns over the benchmark market portfolio when two different types of returns are deployed as the reference point $r$ in solving \textbf{M1} (MSD) and \textbf{M2} (MWSD) problems for each estimation window over the sample period 1964--2024: in (a) the geometric mean of monthly T-bill risk-free rates $(r=rf)$, and (b) the median of benchmark returns $(r=md)$. In turn, the two lower panels report the number of industries $(m=49)$ included in their resulting benchmark dominating portfolios by MSD and MWSD: $X^{rf}$ and $X_{0.18}^{rf}$ in (c), and $X^{md}$ and $X_{0.18}^{md}$ in (d), respectively, for each year from 1964 to 2024. Figure \ref{fig2} quantifies the popularity measured as the frequency (\%) of each industry $j \in \{1,\dots,49\}$ being represented in the dominant MSD and MWSD portfolios $X^{rf}$ and $X_{0.18}^{rf}$, as well as $X^{md}$ and $X_{0.18}^{md}$ across the period 1964--2024. In particular, Figure \ref{fig3}(a) and (b) highlight the most and least preferred industries, respectively, in portfolios $X^{rf}$, $X_{0.18}^{rf}$, $X^{md}$, and $X_{0.18}^{md}$. Moreover, Figure \ref{fig4} demonstrates the `min--max' range of the asset weights $\lambda_j \in [0,1]$, as well as their average values, for each of the $m=49$ industries in each of these dominant portfolios over the period 1964--2024 (see Appendix \ref{appx:heatmaps} for detailed visualizations of asset compositions). A complementary out-of-sample assessment of these benchmark dominating portfolios for the period 1967--2024 is presented in Table \ref{tab:outofsample} and Figure \ref{fig5}.

We first find that the benchmark market portfolio is always MSD-inefficient and MWSD-inefficient as well, since another portfolio can be constructed from the base assets (industry portfolios) that stochastically dominates the benchmark by MSD and MWSD, yielding an excess return of roughly 0.75\% -- 1.75 \% in most 3-year periods. A remarkable exception arises when the 3-year estimation window rolls into the turn of the millennium, during which the benchmark market portfolio becomes considerably inefficient in the sense of MSD and MWSD. For instance, MSD and MWSD investors can beat the market portfolio by approximately 3.5\% during episodes of financial turmoil such as the Dot-com Bubble (2000) and the September 11 Attacks (2001). We also note that using different reference points, namely, the risk-free rate, which is always non-negative, or the median of benchmark returns, which fluctuates between negative and positive values but ensures a 50–50 distribution governing losses and gains in the benchmark, has little impact on the observed excess returns (see Figure \ref{fig1}(a) and (b)). This result confirms that the benchmark market portfolio is indeed MSD- and MWSD-inefficient, regardless of whether MSD and MWSD investors perceive gains and losses relative to the risk-free rate $(r=rf)$ or the benchmark median $(r=\text{median}(Y))$.

\begin{figure}[h!]
\includegraphics[width=\textwidth, keepaspectratio]{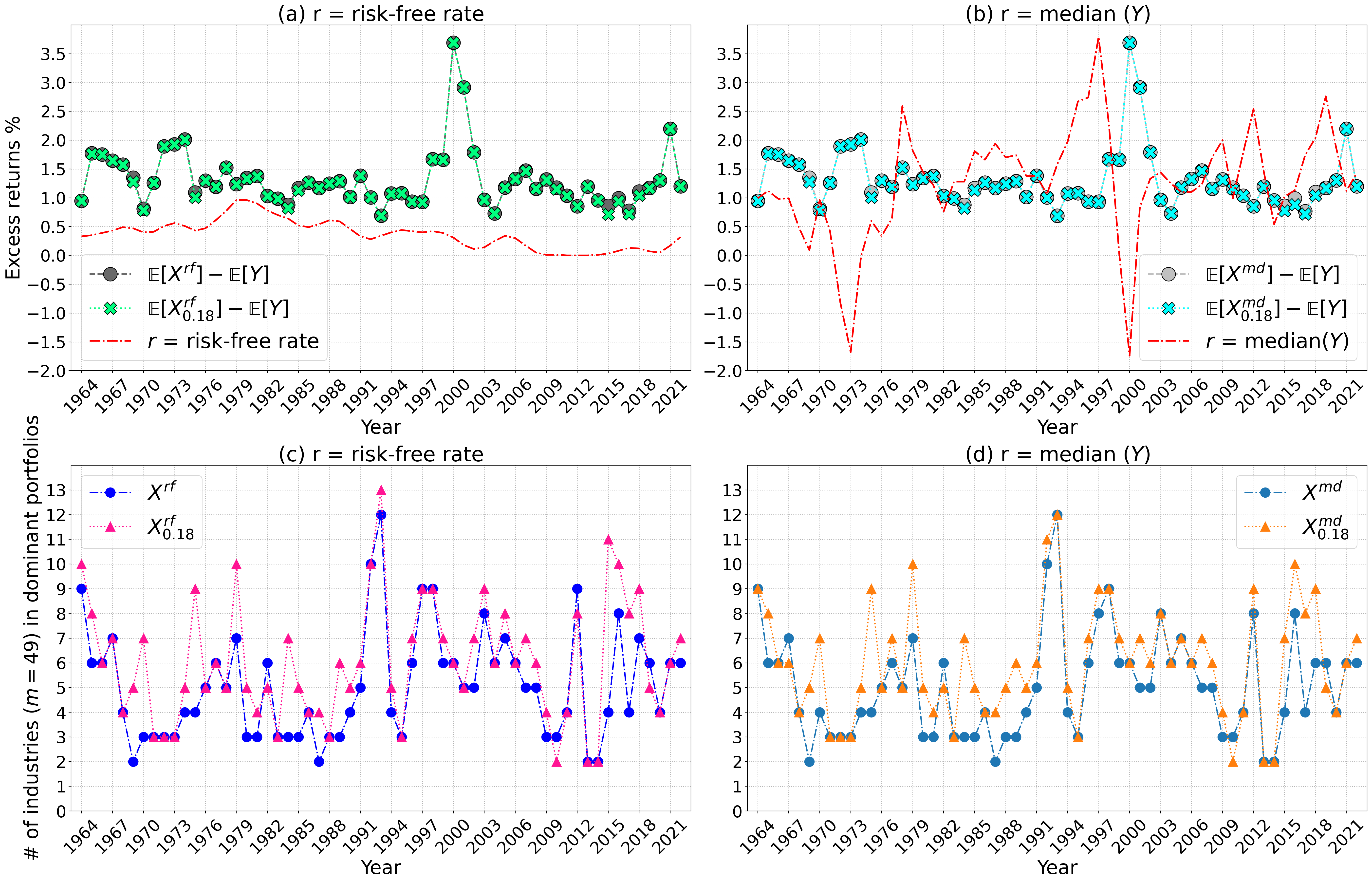}
\caption{Inefficiency of the market portfolio and number of industries (out of 49) included in MSD portfolios $X^{rf}$, $X^{md}$ and MWSD portfolios $X_{0.18}^{rf}$, $X_{0.18}^{md}$, 1964--2024.}\label{fig1}
\end{figure}

Next, we document that in most 3-year periods MWSD investors who distort the state probabilities of gains and losses earn portfolio returns nearly as great as those of MSD investors. MWSD returns are only very marginally lower in some periods (see Figure \ref{fig1}(a) and (b)). Thus, allowing sets of feasible inverse S-shaped probability weighting functions (PWFs) results in trivial return underperformance for MWSD investors as compared to MSD investors. Moreover, we then find that MSD and MWSD investors do diversify across the asset universe rather than allocating all capital entirely into one base asset. In particular, MSD and MWSD investors diversify their portfolios by investing into 2--13 out of the $m=49$ industries (see Figure \ref{fig1}(c) and (d)) regardless of the reference point applied. Surprisingly, MSD and MWSD investors appear to hold largely different portfolios in quite many 3-year periods, although both achieve almost equally on portfolio returns. In general, MWSD investors have to diversify more broadly across the underlying asset universe to accommodate probability distortion. The required diversification is considerably broader during those 3-year periods, especially, when MWSD investors earn slightly lower portfolio returns than their MSD counterparts.

Then, we observe explicit asset preferences to certain industries among MSD and MWSD investors from their portfolio asset compositions (see Appendix \ref{appx:heatmaps} for detailed visualizations of asset compositions). Clearly, most of the $m=49$ industries have secured their asset representation in MSD and MWSD portfolios (see Figure \ref{fig2}). In particular, the most preferred industries are $\lceil \text{Soda} \rfloor$, $\lceil \text{Smoke} \rfloor$, $\lceil \text{Guns} \rfloor$, $\lceil \text{Gold} \rfloor$, $\lceil \text{Coal} \rfloor$, and $\lceil \text{Oil} \rfloor$. Each of these industries attracts no less than 12 asset allocations over a total of 59 3-year periods (see Figure \ref{fig3}(a)). Interestingly, $\lceil \text{Gold} \rfloor$ surfaces as the most popular industry among others for both MSD and MWSD investors. The underlying intuition is arguably that MSD and MWSD investors are featured with risk aversion over losses, while gold is typically considered a stable safe-haven asset whose price positively correlates with risk aversion. Nevertheless, the asset weight of $\lceil \text{Gold} \rfloor$ has never exceeded 50\% in any MSD portfolios $X^{rf}$, $X^{md}$ or any MWSD portfolios $X_{0.18}^{rf}$, $X_{0.18}^{md}$ throughout the period 1964--2024 (see Figure \ref{fig4}). Among the most preferred industries, we also note that each of them is typically more popular among MWSD investors than MSD investors. Such an empirical observation aligns consistently with our finding that MWSD investors diversify more broadly than MSD counterparts. On the other hand, however, the least preferred industries are $\lceil \text{Chems} \rfloor$, $\lceil \text{Rubbr} \rfloor$, $\lceil \text{BldMt} \rfloor$, $\lceil \text{BusSv} \rfloor$, $\lceil \text{LabEq} \rfloor$, $\lceil \text{Paper} \rfloor$, as well as $\lceil \text{Whlsl} \rfloor$. Each of these industries receives no more than 1 asset allocation over 59 3-year periods (see Figure \ref{fig3}(b)). Notably, three industries $\lceil \text{Rubbr} \rfloor$, $\lceil \text{BldMt} \rfloor$, and $\lceil \text{Whlsl} \rfloor$ are not included in any MSD portfolios $X^{rf}$, $X^{md}$ or any MWSD portfolios $X_{0.18}^{rf}$, $X_{0.18}^{md}$, while only $\lceil \text{BusSv} \rfloor$ is never included in MSD portfolios $X^{rf}$, $X^{md}$ throughout the period 1964--2024.

Finally, to complement the in-sample analysis presented thus far, we also report the out-of-sample performance of MSD portfolios $X^{rf}$, $X^{md}$ and MWSD portfolios $X_{0.18}^{rf}$, $X_{0.18}^{md}$ and compare their performance with that of the benchmark market portfolio over 58 non-overlapping 1-year periods spanning Jan/1967 -- Dec/1967 to Jan/2024 -- Dec/2024 (see Table \ref{tab:outofsample}). Surprisingly, behavioral investors with MSD and MWSD preferences do not, on average, substantially underperform the benchmark market portfolio out-of-sample. Nevertheless, both MSD portfolios $X^{rf}$, $X^{md}$ and MWSD portfolios $X_{0.18}^{rf}$, $X_{0.18}^{md}$ are moderately less risky than the benchmark market portfolio in terms of downside risk. This finding is supported by their lower expected shortfall (CVaR$_{5\%}$) and marginally higher risk-adjusted Sortino ratios. Interestingly, MSD and MWSD investors can outperform the benchmark market portfolio in cumulative portfolio growth over a sufficiently long investment horizon (see Figure \ref{fig5}).

\begin{table}[h!]
\centering
\begin{small}
\caption{Annualized performance of MSD portfolios $X^{rf}$, $X^{md}$ (out-of-sample), MWSD portfolios $X_{0.18}^{rf}$, $X_{0.18}^{md}$ (out-of-sample), and benchmark market portfolio $Y$, 1967--2024. \label{tab:outofsample}}
\begin{tabular}{lllllll}
\hline\noalign{\smallskip}
Portfolio & & Mean & Std. & CVaR$_{5\%}$ & Sharpe & Sortino \\
\noalign{\smallskip}
          & & \%   & \%   & \%  & ratio & ratio \\ 
\noalign{\smallskip}
\hline
\noalign{\smallskip}
$X^{rf}$ & & 12.22 & 18.27 & -26.93 & 0.419 & 0.784 \\
\noalign{\smallskip}
$X_{0.18}^{rf}$ & & 12.22 & 17.85 & -25.23 & 0.430 &  0.802 \\
\noalign{\smallskip}
$X^{md}$ & & 12.17 & 18.22 & -26.93 & 0.418 & 0.778 \\
\noalign{\smallskip}
$X_{0.18}^{md}$ & & 12.18 & 17.75 & -25.18 & 0.431 & 0.805 \\
\noalign{\smallskip}
$Y$ & & 12.32 & 17.62 & -28.54 & 0.435 & 0.764 \\
\noalign{\smallskip}
\hline
\end{tabular}
\vspace{0.5em}
\parbox{0.95\linewidth}{\footnotesize
Note: Sharpe and Sortino ratios are computed as the mean of the annualized portfolio returns in excess of the annualized risk-free rates, divided by the standard deviation and downside deviation of the annualized excess returns, respectively.}
\end{small}
\end{table}

\begin{figure}[h!]
\includegraphics[width=\textwidth, keepaspectratio]{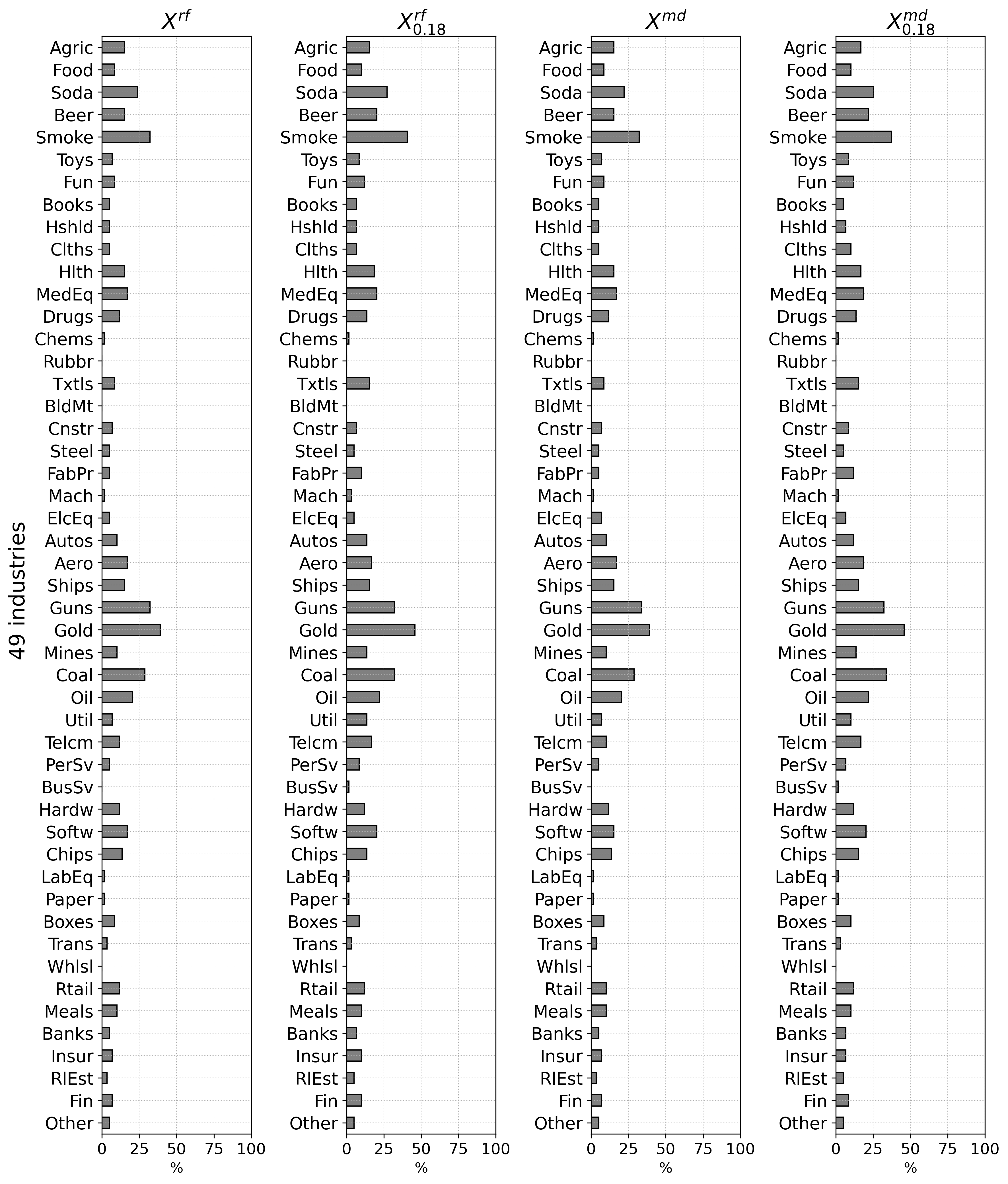}
\caption{Popularity of the 49 industries in MSD portfolios $X^{rf}$, $X^{md}$ and MWSD portfolios $X_{0.18}^{rf}$, $X_{0.18}^{md}$.}\label{fig2}
\end{figure}

\begin{figure}[h!]
\includegraphics[width=\textwidth, keepaspectratio]{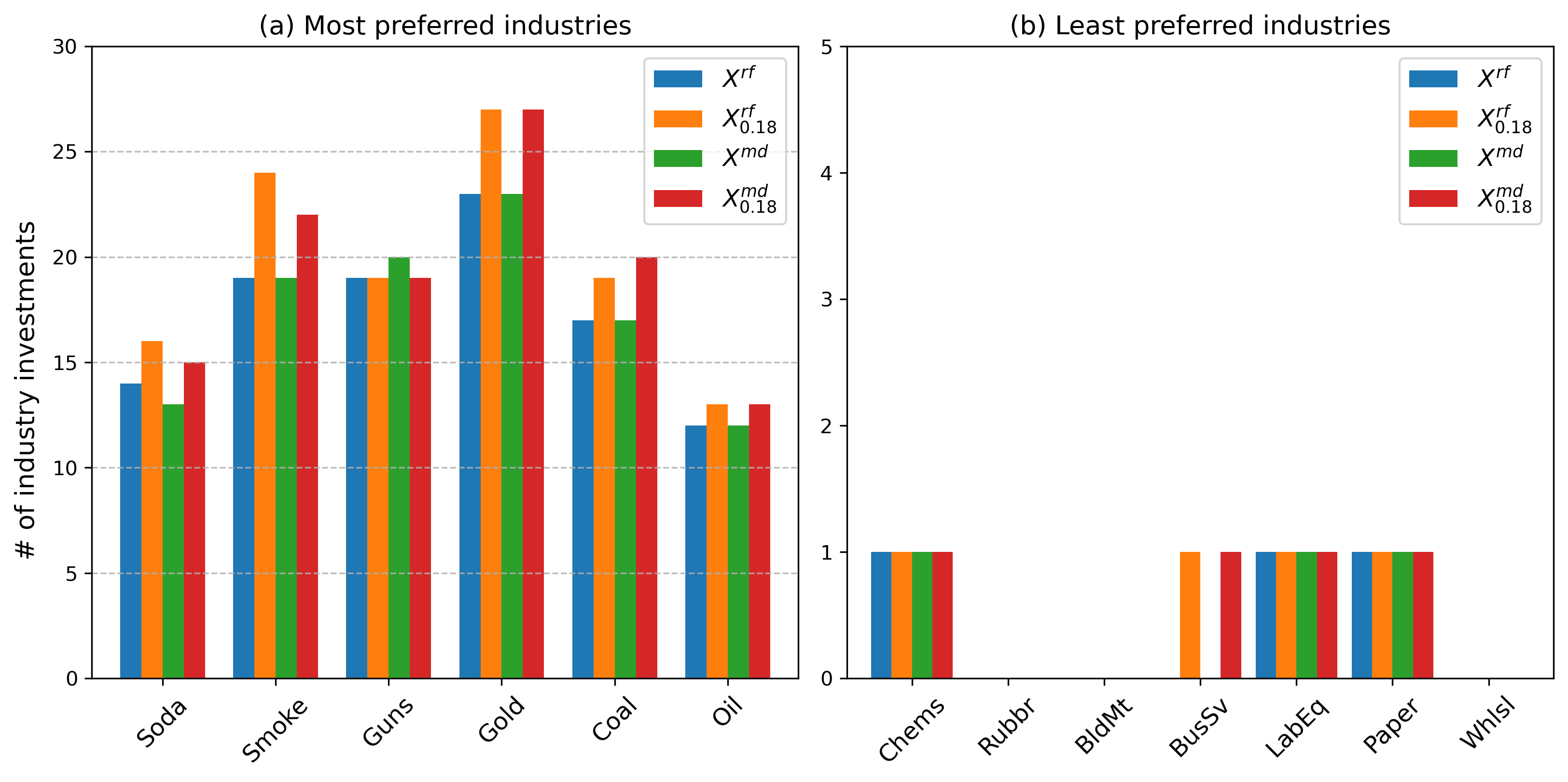}
\caption{Most and least preferred industries in MSD portfolios $X^{rf}$, $X^{md}$ and MWSD portfolios $X_{0.18}^{rf}$, $X_{0.18}^{md}$.}\label{fig3}
\end{figure}

\begin{figure}[h!]
\includegraphics[width=\textwidth, keepaspectratio]{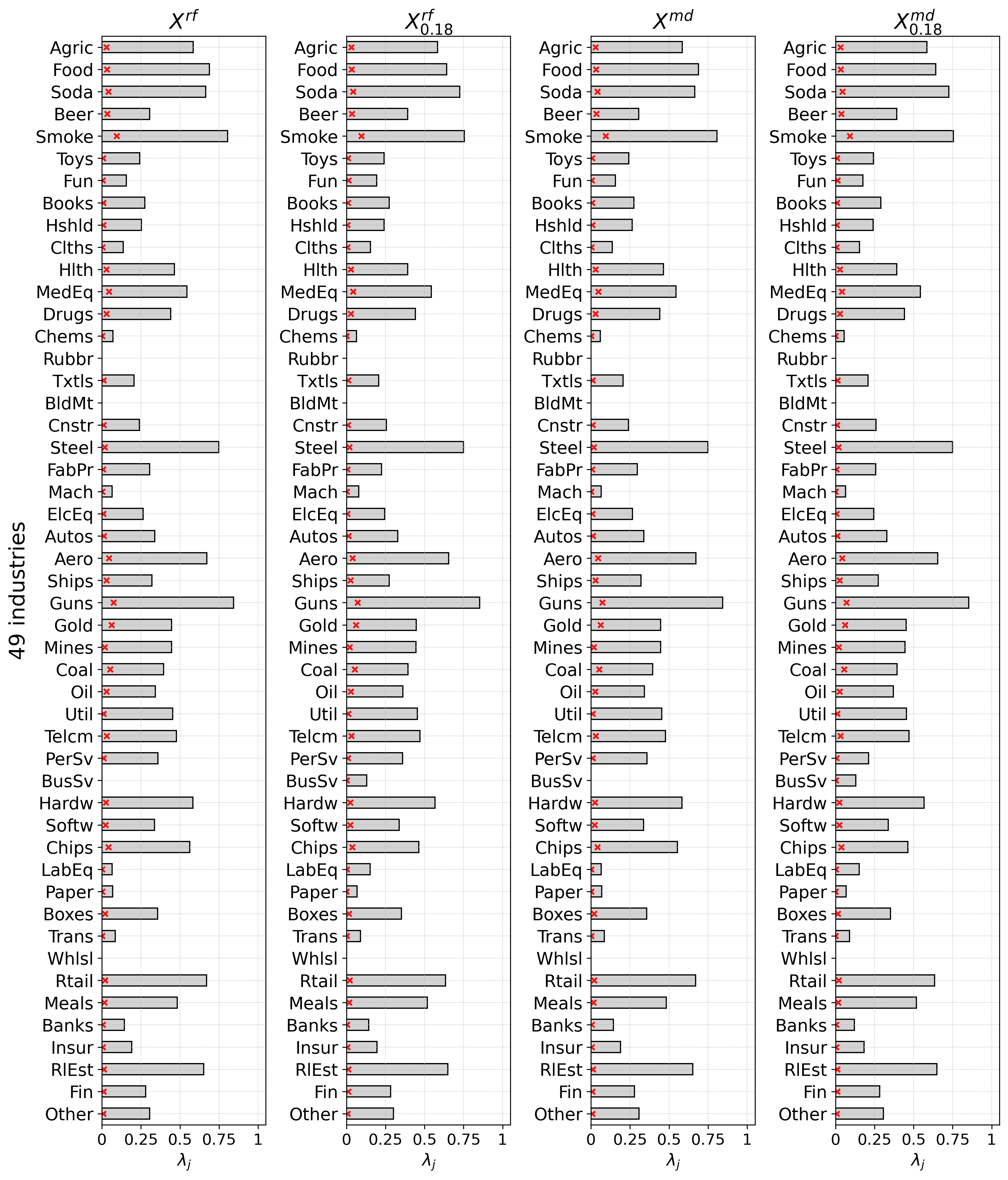}
\caption{Asset compositions of MSD portfolios $X^{rf}$, $X^{md}$ and MWSD portfolios $X_{0.18}^{rf}$, $X_{0.18}^{md}$. The gray bars indicate the `min--max' range of asset weights $\lambda_j \in [0,1]$ and the red crosses denote their mean values for $j \in \{1,\dots,49\}$.}\label{fig4}
\end{figure}

\begin{figure}[h!]
\includegraphics[width=\linewidth, keepaspectratio]{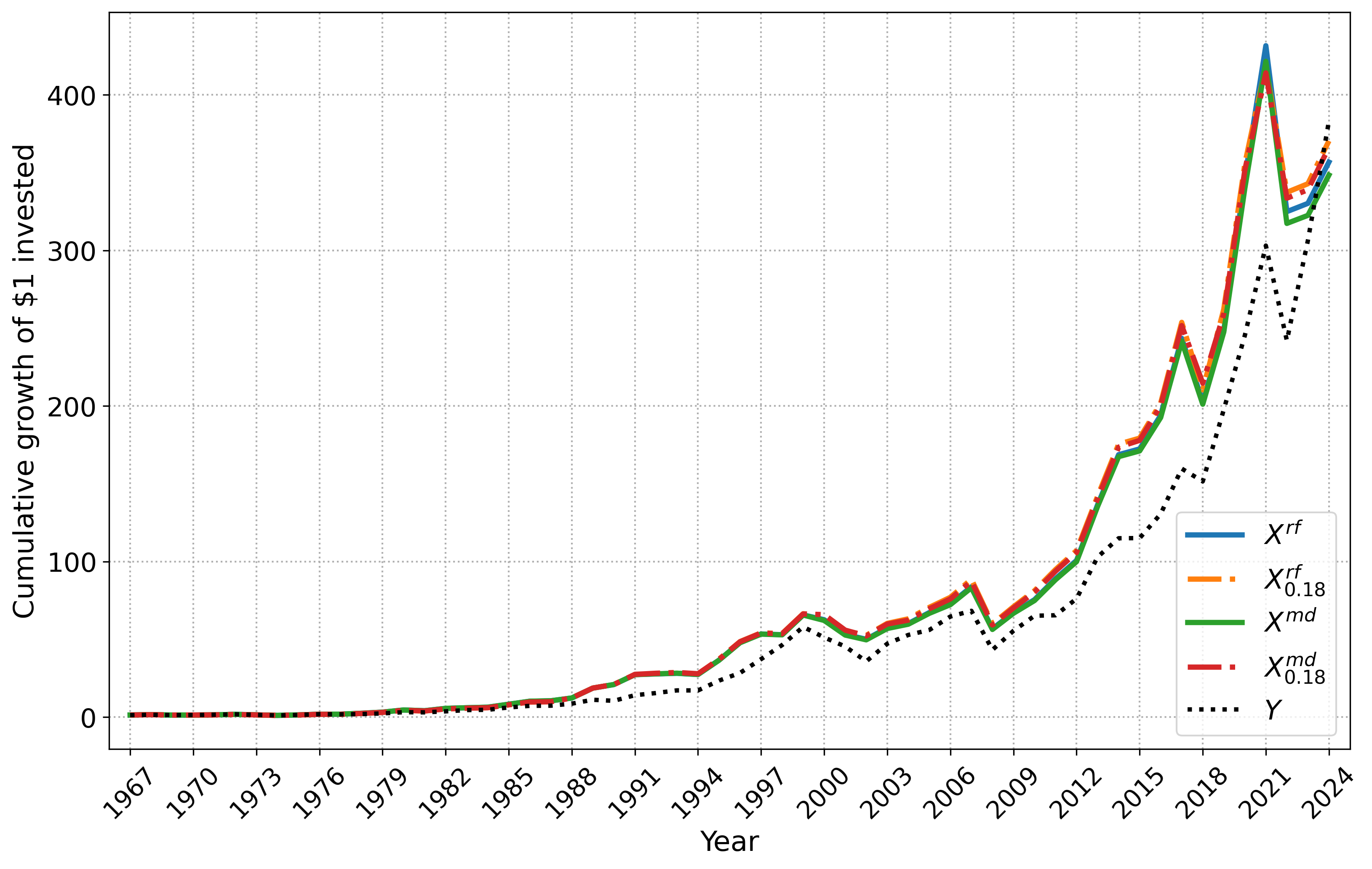}
\caption{Cumulative growth of \$1 invested in MSD portfolios $X^{rf}$, $X^{md}$, MWSD portfolios $X_{0.18}^{rf}$, $X_{0.18}^{md}$, and benchmark market portfolio $Y$, obtained by compounding annual portfolio returns over 1967--2024.} \label{fig5}
\end{figure}

\section{Conclusions and Avenues of Future Research} \label{sec6}

In this paper, we have developed novel stochastic portfolio optimization models for two distinct Markowitz Stochastic Dominance (MSD) criteria allowing incompletely or imprecisely specified risk preferences. MSD preferences are typically characterized by an inverse S-shaped utility function and a pair of inverse S-shaped probability weighting functions, both of which often require explicit and complete non-linear specifications. However, we adopt a different approach by admitting such incomplete MSD preferences harnessing the MSD and MWSD conditions established in a discrete setting. We have further demonstrated that these dominance conditions can be mathematically expressed in the form of a system of linear constraints and their resulting portfolio optimization models can be formulated as mixed-integer linear programming (MILP) problems. The developed portfolio optimization models enable us to identify an optimal portfolio that is preferred to a pre-specified benchmark return distribution by all behavioral investors whose portfolio choice aligns with MSD or MWSD preferences.

Furthermore, our empirical analysis reveals that the benchmark market portfolio is MSD- and MWSD-inefficient and thus cannot be justified as an optimal investment alternative for behavioral investors with MSD and MWSD preferences. Moreover, classic behavioral factors such as the choice of reference point (risk-free rate or benchmark median) and probability distortion report little to no influence on MSD and MWSD portfolio returns. However, notable distinctions arise between these two behavioral investor classes, as empirically observed in their portfolio choices with respect to diversification patterns and asset compositions. Finally, both MSD and MWSD investors exhibit explicit asset preferences to particular industries, where $\lceil \text{Gold} \rfloor$ emerging as the most favored one among them.

This paper also opens up several avenues for future research. In terms of theoretical development, the next natural step is to establish MWSD conditions that explicitly incorporate another behavioral factor, i.e., loss aversion, on top of probability distortion. Prior empirical studies have found that investors distort probabilities less when a loss is perceived than a gain. Building on this evidence, future theoretical work could develop dominance conditions that accommodate all inverse S-shaped utilities with steeper slopes in losses than in gains, as well as all inverse S-shaped probability weighting functions (PWFs) whose slopes are greater in the loss domain, i.e., $\frac{\partial}{\partial t}w^-(\cdot) > \frac{\partial}{\partial t}w^+(\cdot)$, to allow asymmetric probability weighting. Methodologically, new stochastic optimization models that utilize these novel dominance conditions can subsequently be developed to better characterize and analyze the portfolio choice of behavioral investors with such explicit preferences.

Although the proposed MSD- and MWSD-constrained optimization models are not intended to prescribe rational portfolio decisions as EUT-based normative decision models do, they provide new analytical and computational tools for fund and portfolio managers to extract insights from actual portfolio decisions. Specifically, financial practitioners can utilize these advances in portfolio optimization to evaluate and justify whether clients' portfolio choices are consistent with specific behavioral preferences. In a recent work, \cite{XuLiesio2026} develop stochastic optimization models for incomplete Cumulative Prospect Theory (CPT; \citealt{TverskyKahneman1992}) preferences, which represent the opposite of MSD and MWSD preferences. Their results show that CPT investors with a pair of convex PWFs do not adequately diversify across the available assets, but instead occasionally allocate their entire wealth to just one single asset. Taken together, the findings of the present study and those of \cite{XuLiesio2026} illustrate how behavioral optimization models can reveal systematic patterns in the observed portfolio choices of certain investor groups that cannot be accommodated by normative portfolio models alone.

Consequently, one important direction for future research is to apply both normative and behavioral portfolio optimization models to global asset markets, such as cryptocurrencies, using daily or weekly return data, across different sample periods, and in decision settings with other benchmark alternatives or without explicit benchmarks (see, e.g., \citealt{Longarela2016, LiesioEtAl2023}). Comparing the portfolios generated by these complementary paradigms would enable financial practitioners to provide more personalized portfolio recommendations, improve investor profiling, and construct portfolios that appropriately balance rational asset allocations with investors' behavioral preferences and quantify the associated trade-offs between them.

Beyond financial portfolio optimization, the stochastic optimization models developed in this work have broader practical applicability and are well-suited to behavioral decision analysis across several fields such as behavioral finance and economics, behavioral operations and supply chain management. For instance, another promising direction for future research is to apply the developed models in an operations management context to examine and analyze forest harvesting decisions (\citealt{Heikkinen2003, Kuosmanen2004, LejeuneKettunen2017, SinhaEtAl2017}). In such a setting, the base assets would correspond to forestry assets of harvestable stands of tree species, saw-logs or pulpwood (see, e.g., \citealt{Kuosmanen2004}). The optimization models here can then be used to identify if there exists another return or cost distribution that dominates a benchmark in the sense of MSD or MWSD, and a forest or landowner's harvesting routine can serve as the benchmark to assess whether MSD or MWSD justifies the current operational practice.

\onehalfspacing
\setlength\bibsep{0pt}
\bibliographystyle{apalike}
\bibliography{Ref}

\begin{thebibliography}{}

\bibitem[Allais, 1953]{Allais1953}
Allais, M. (1953).
\newblock Le comportement de l'homme rationnel devant le risque: Critique des postulats et axiomes de l'ecole americaine.
\newblock {\em Econometrica}, 21(4):503--546.

\bibitem[Arvanitis et~al., 2019]{ArvanitisEtAl2019}
Arvanitis, S., Hallam, M., Post, T., and Topaloglou, N. (2019).
\newblock Stochastic spanning.
\newblock {\em Journal of Business \& Economic Statistics}, 37(4):573--585.

\bibitem[Arvanitis et~al., 2020]{ArvanitisEtAl2020}
Arvanitis, S., Scaillet, O., and Topaloglou, N. (2020).
\newblock Spanning tests for {M}arkowitz stochastic dominance.
\newblock {\em Journal of Econometrics}, 217(2):291--311.

\bibitem[Arvanitis and Topaloglou, 2017]{ArvanitisTopaloglou2017}
Arvanitis, S. and Topaloglou, N. (2017).
\newblock Testing for prospect and {M}arkowitz stochastic dominance efficiency.
\newblock {\em Journal of Econometrics}, 198(2):253--270.

\bibitem[Baucells and Heukamp, 2006]{BaucellsHeukamp2006}
Baucells, M. and Heukamp, F.~H. (2006).
\newblock Stochastic dominance and cumulative prospect theory.
\newblock {\em Management Science}, 52(9):1409--1423.

\bibitem[Bertsimas and Kallus, 2020]{BertsimasKallus2020}
Bertsimas, D. and Kallus, N. (2020).
\newblock From predictive to prescriptive analytics.
\newblock {\em Management Science}, 66(3):1025--1044.

\bibitem[Bruni et~al., 2017]{BruniEtAl2017}
Bruni, R., Cesarone, F., Scozzari, A., and Tardella, F. (2017).
\newblock On exact and approximate stochastic dominance strategies for portfolio selection.
\newblock {\em European Journal of Operational Research}, 259(1):322--329.

\bibitem[Cesarone et~al., 2026]{CesaroneEtAl2026}
Cesarone, F., Corradini, M., Lampariello, L., and Riccioni, J. (2026).
\newblock A new behavioral model for portfolio selection using the half-full/half-empty approach.
\newblock {\em European Journal of Operational Research}, 330(2):687--699.

\bibitem[Cesarone and Puerto, 2025]{CesaronePuerto2025}
Cesarone, F. and Puerto, J. (2025).
\newblock Flexible enhanced indexation models through stochastic dominance and ordered weighted average optimization.
\newblock {\em European Journal of Operational Research}, 323(2):657--670.

\bibitem[Cinfrignini et~al., 2025]{CinfrigniniEtAl2025}
Cinfrignini, A., Petturiti, D., and Vantaggi, B. (2025).
\newblock Behavioral dynamic portfolio selection with {S}-shaped utility and epsilon-contaminations.
\newblock {\em European Journal of Operational Research}, 325(3):500--515.

\bibitem[Dentcheva and Ruszczy\'{n}ski, 2003]{DentchevaRuszczynski2003}
Dentcheva, D. and Ruszczy\'{n}ski, A. (2003).
\newblock Optimization with stochastic dominance constraints.
\newblock {\em SIAM Journal on Optimization}, 14(2):548--566.

\bibitem[Dentcheva and Ruszczy\'{n}ski, 2006]{DentchevaRuszczynski2006}
Dentcheva, D. and Ruszczy\'{n}ski, A. (2006).
\newblock Portfolio optimization with stochastic dominance constraints.
\newblock {\em Journal of Banking and Finance}, 30(2):433--451.

\bibitem[Dentcheva and Ruszczy\'{n}ski, 2010]{DentchevaRuszczynski2010}
Dentcheva, D. and Ruszczy\'{n}ski, A. (2010).
\newblock Robust stochastic dominance and its application to risk-averse optimization.
\newblock {\em Mathematical Programming}, 123:85--100.

\bibitem[Dupa\v{c}ov\'{a} and Kopa, 2012]{DupacovaKopa2012}
Dupa\v{c}ov\'{a}, J. and Kopa, M. (2012).
\newblock Robustness in stochastic programs with risk constraints.
\newblock {\em Annals of Operations Research}, 200:55--74.

\bibitem[Dupa\v{c}ov\'{a} and Kopa, 2014]{DupacovaKopa2014}
Dupa\v{c}ov\'{a}, J. and Kopa, M. (2014).
\newblock Robustness of optimal portfolios under risk and stochastic dominance constraints.
\newblock {\em European Journal of Operational Research}, 234:434--441.

\bibitem[Egozcue et~al., 2011]{EgozcueEtAl2011}
Egozcue, M., Garc\'{i}a, L.~F., Wong, W.-K., and Zitikis, R. (2011).
\newblock Do investors like to diversify? {A} study of {M}arkowitz preferences.
\newblock {\em European Journal of Operational Research}, 215(1):188--193.

\bibitem[Gustafsson and Salo, 2005]{GustafssonSalo2005}
Gustafsson, J. and Salo, A. (2005).
\newblock Contingent portfolio programming for the management of risky projects.
\newblock {\em Operations Research}, 53(6):946--956.

\bibitem[Gutjahr, 2015]{Gutjahr2015}
Gutjahr, W.~J. (2015).
\newblock Bi-objective multi-mode project scheduling under risk aversion.
\newblock {\em European Journal of Operational Research}, 246(2):421--434.

\bibitem[Hadar and Russell, 1969]{HadarRussell1969}
Hadar, J. and Russell, W. (1969).
\newblock Rules for ordering uncertain prospects.
\newblock {\em American Economic Review}, 59(1):25--34.

\bibitem[Hanoch and Levy, 1969]{HanochLevy1969}
Hanoch, G. and Levy, H. (1969).
\newblock The efficiency analysis of choices involving risk.
\newblock {\em The Review of Economic Studies}, 36(3):335--346.

\bibitem[Harris and Mazibas, 2022]{HarrisMazibas2022}
Harris, R.~D. and Mazibas, M. (2022).
\newblock Portfolio optimization with behavioural preferences and investor memory.
\newblock {\em European Journal of Operational Research}, 296(1):368--387.

\bibitem[Heikkinen, 2003]{Heikkinen2003}
Heikkinen, V.-P. (2003).
\newblock Timber harvesting as a part of the portfolio management: A multiperiod stochastic optimisation approach.
\newblock {\em Management Science}, 49(1):131--142.

\bibitem[Kahneman and Tversky, 1979]{KahnemanTversky1979}
Kahneman, D. and Tversky, A. (1979).
\newblock Prospect theory: An analysis of decision under risk.
\newblock {\em Econometrica}, 47(2):263--291.

\bibitem[Kallio and Hardoroudi, 2018]{KallioHardoroudi2018}
Kallio, M. and Hardoroudi, N.~D. (2018).
\newblock Second-order stochastic dominance constrained portfolio optimization: Theory and computational tests.
\newblock {\em European Journal of Operational Research}, 264:675--685.

\bibitem[Kofina et~al., 2025]{KofinaEtAl2025}
Kofina, A., Psaradellis, I., and Topaloglou, N. (2025).
\newblock Revisiting {M}arkowitz stochastic dominance in international markets.
\newblock {\em The European Journal of Finance}, pages 1--24.

\bibitem[Kopa and Kozm\'{i}k, 2025]{KopaKozmik2025}
Kopa, M. and Kozm\'{i}k, K. (2025).
\newblock Robust approaches in portfolio optimization with stochastic dominance constraints.
\newblock {\em OR Spectrum}.

\bibitem[Kopa et~al., 2018]{KopaEtAl2018}
Kopa, M., Moriggia, V., and Vitali, S. (2018).
\newblock Individual optimal pension allocation under stochastic dominance constraints.
\newblock {\em Annals of Operations Research}, 260:255--291.

\bibitem[Kopa and Post, 2009]{KopaPost2009}
Kopa, M. and Post, T. (2009).
\newblock A portfolio optimality test based on the first-order stochastic dominance criterion.
\newblock {\em Journal of Financial and Quantitative Analysis}, 44(5):1103--1124.

\bibitem[Kopa and Post, 2015]{KopaPost2015}
Kopa, M. and Post, T. (2015).
\newblock A general test for {SSD} portfolio efficiency.
\newblock {\em OR Spectrum}, 37:703--734.

\bibitem[Kopa et~al., 2027]{KopaEtAl2027}
Kopa, M., Post, T., and Junov\'{a}, J. (2027).
\newblock Portfolio optimization with decreasing absolute risk aversion stochastic dominance constraints.
\newblock {\em European Journal of Operational Research}, 336(2):794--807.

\bibitem[Kuosmanen, 2004]{Kuosmanen2004}
Kuosmanen, T. (2004).
\newblock Efficient diversification according to stochastic dominance criteria.
\newblock {\em Management Science}, 50(10):1390--1406.

\bibitem[Lejeune and Kettunen, 2017]{LejeuneKettunen2017}
Lejeune, M.~A. and Kettunen, J. (2017).
\newblock Managing reliability and stability risks in forest harvesting.
\newblock {\em Manufacturing and Service Operations Management}, 19(4):620--638.

\bibitem[Levy, 2016]{Levy2016}
Levy, H. (2016).
\newblock {\em Stochastic Dominance: Investment Decision Making under Uncertainty}.
\newblock Springer, 3rd edition.

\bibitem[Levy and Levy, 2002]{LevyLevy2002}
Levy, M. and Levy, H. (2002).
\newblock Prospect theory: Much ado about nothing?
\newblock {\em Management Science}, 48(10):1334--1349.

\bibitem[Liesi{\"o} et~al., 2023]{LiesioEtAl2023}
Liesi{\"o}, J., Kallio, M., and Argyris, N. (2023).
\newblock Incomplete risk-preference information in portfolio decision analysis.
\newblock {\em European Journal of Operational Research}, 304(3):1084--1098.

\bibitem[Liesi{\"o} and Salo, 2012]{LiesioSalo2012}
Liesi{\"o}, J. and Salo, A. (2012).
\newblock Scenario-based portfolio selection of investment projects with incomplete probability and utility information.
\newblock {\em European Journal of Operational Research}, 217:162--172.

\bibitem[Liesi{\"o} et~al., 2021]{LiesioEtAl2021}
Liesi{\"o}, J., Salo, A., Keisler, J.~M., and Morton, A. (2021).
\newblock Portfolio decision analysis: Recent developments and future prospects.
\newblock {\em European Journal of Operational Research}, 293(3):811--825.

\bibitem[Liesi{\"o} et~al., 2020]{LiesioEtAl2020}
Liesi{\"o}, J., Xu, P., and Kuosmanen, T. (2020).
\newblock Portfolio diversification based on stochastic dominance under incomplete probability information.
\newblock {\em European Journal of Operational Research}, 286(2):755--768.

\bibitem[Longarela, 2016]{Longarela2016}
Longarela, I.~R. (2016).
\newblock A characterization of the {SSD}-efficient frontier of portfolio weights by means of a set of mixed-integer linear constraints.
\newblock {\em Management Science}, 62(12):3549--3554.

\bibitem[Markowitz, 1952]{MarkowitzSD1952}
Markowitz, H. (1952).
\newblock The utility of wealth.
\newblock {\em Journal of Political Economy}, 60(2):151--158.

\bibitem[Moskowitz et~al., 1993]{Moskowitz1993}
Moskowitz, H., Preckel, P.~V., and Yang, A. (1993).
\newblock Decision analysis with incomplete utility and probability information.
\newblock {\em Operations Research}, 41(5):864--879.

\bibitem[Neugebauer, 2026]{Neugebauer2026}
Neugebauer, J. (2026).
\newblock Portfolio optimization with robust stochastic dominance testing: A genetic algorithm approach.
\newblock {\em European Journal of Operational Research}, 333(2):519--533.

\bibitem[Neumann and Morgenstern, 1944]{NeumannMorgenstern1944}
Neumann, J.~v. and Morgenstern, O. (1944).
\newblock {\em Theory of Games and Economic Behavior}.
\newblock Princeton University Press.

\bibitem[Peng and Delage, 2024]{PengDelage2024}
Peng, C. and Delage, E. (2024).
\newblock Data-driven optimization with distributionally robust second order stochastic dominance constraints.
\newblock {\em Operations Research}, 72(3):1298--1316.

\bibitem[Post, 2003]{Post2003}
Post, T. (2003).
\newblock Empirical tests for stochastic dominance efficiency.
\newblock {\em Journal of Finance}, 58(5):1905--1931.

\bibitem[Post et~al., 2018]{PostEtAl2018}
Post, T., Karabati, S., and Arvanitis, S. (2018).
\newblock Portfolio optimization based on stochastic dominance and empirical likelihood.
\newblock {\em Journal of Econometrics}, 206:167--186.

\bibitem[Post and Kopa, 2017]{PostKopa2017}
Post, T. and Kopa, M. (2017).
\newblock Portfolio choice based on third-degree stochastic dominance.
\newblock {\em Management Science}, 63(10):3381--3392.

\bibitem[Post and Levy, 2005]{PostLevy2005}
Post, T. and Levy, H. (2005).
\newblock Does risk seeking drive stock prices? {A} stochastic dominance analysis of aggregate investor preferences and beliefs.
\newblock {\em Review of Financial Studies}, 18(3):925--953.

\bibitem[Quirk and Saposnik, 1962]{QuirkSaposnik1962}
Quirk, J.~P. and Saposnik, R. (1962).
\newblock Admissibility and measurable utility functions.
\newblock {\em The Review of Economic Studies}, 29(2):140--146.

\bibitem[Rothschild and Stiglitz, 1970]{RothschildStiglitz1970}
Rothschild, M. and Stiglitz, J.~E. (1970).
\newblock Increasing risk: I. {A} definition.
\newblock {\em Journal of Economic Theory}, 2:225--243.

\bibitem[Salo et~al., 2011]{SaloEtAl2011}
Salo, A., Keisler, J., and Morton, A. (2011).
\newblock {\em Advances in Portfolio Decision Analysis: Improved Methods for Resource Allocation}.
\newblock Springer, New York.

\bibitem[Sinha et~al., 2017]{SinhaEtAl2017}
Sinha, A., Rämö, J., Malo, P., Kallio, M., and Tahvonen, O. (2017).
\newblock Optimal management of naturally regenerating uneven-aged forests.
\newblock {\em European Journal of Operational Research}, 256(3):886--900.

\bibitem[Starmer, 2000]{Starmer2000}
Starmer, C. (2000).
\newblock Developments in non-expected utility theory: The hunt for a descriptive theory of choice under risk.
\newblock {\em Journal of Economic Literature}, 38(2):332--382.

\bibitem[Tversky and Kahneman, 1992]{TverskyKahneman1992}
Tversky, A. and Kahneman, D. (1992).
\newblock Advances in prospect theory: Cumulative representation of uncertainty.
\newblock {\em Journal of Risk and Uncertainty}, 5:297--323.

\bibitem[Vilkkumaa et~al., 2018]{VilkkumaaEtAl2018}
Vilkkumaa, E., Liesi{\"o}, J., Salo, A., and Ilmola-Sheppard, L. (2018).
\newblock Scenario-based portfolio model for building robust and proactive strategies.
\newblock {\em European Journal of Operational Research}, 266(1):205--220.

\bibitem[Wakker, 2010]{Wakker2010}
Wakker, P.~P. (2010).
\newblock {\em Prospect theory: For risk and ambiguity}.
\newblock Cambridge University Press.

\bibitem[Whitmore, 1970]{Whitmore1970}
Whitmore, G. (1970).
\newblock Third-degree stochastic dominance.
\newblock {\em American Economic Review}, 60(3):457--459.

\bibitem[Xu, 2024]{Xu2024}
Xu, P. (2024).
\newblock Testing out-of-sample portfolio performance using second-order stochastic dominance constrained optimization approach.
\newblock {\em International Review of Financial Analysis}, 95:103368.

\bibitem[Xu and Liesi{\"o}, 2026]{XuLiesio2026}
Xu, P. and Liesi{\"o}, J. (2026).
\newblock Optimization models for cumulative prospect theory under incomplete preference information.
\newblock {\em European Journal of Operational Research}, 330(1):217--229.

\end{thebibliography}


\clearpage
\appendix
\section{Proofs}\label{appx:proofs}

\noindent \textbf{Proof of Theorem \ref{th:discrete_MSD}}\label{appx:th_discrete_MSD}

By condition \eqref{eq:msd_gains} of Proposition \ref{prop:msd}, for all $t^+ \in [r, b]$, it follows
\begin{eqnarray*}
& &\int_{t^+}^b \big[F_Y(\tau) - F_X(\tau)\big]d\tau \geq 0 \\
&\Leftrightarrow &\int_{t^+}^b F_Y(\tau)d\tau - \int_{t^+}^b F_X(\tau)d\tau \geq 0 \\
&\Leftrightarrow  &\int_{a}^b F_Y(\tau)d\tau - \int_a^{t^+} F_Y(\tau)d\tau - \left( \int_{a}^b F_X(\tau)d\tau - \int_a^{t^+} F_X(\tau)d\tau \right) \geq 0 \\
&\Leftrightarrow & F^2_Y(b)-F^2_Y(t^+) - \left( F^2_X(b)-F^2_X(t^+) \right) \geq 0 \\
&\Leftrightarrow & F^2_Y(b)-F^2_Y(t^+) \geq F^2_X(b)-F^2_X(t^+),
\end{eqnarray*}
which is the continuous case of condition \eqref{eq:th_gains} of Theorem \ref{th:discrete_MSD}. What remains to be proved is that under a discrete state-space condition \eqref{eq:th_gains} implies this condition. We prove this by contrapositive, i.e., if the continuous case of condition \eqref{eq:th_gains} does not hold, then condition \eqref{eq:th_gains} does not hold either. Let $\Theta(t^+) = F^2_Y(b)-F^2_Y(t^+)$ and $\Psi(t^+) = F^2_X(b)-F^2_X(t^+)$ for all $t^+ \in \{x_i ~|~ x_i > r\}$. Assume now the continuous case of condition \eqref{eq:th_gains} does not hold for some $t^* > r$ in gains, then $\Theta(t^*) - \Psi(t^*) < 0$. First, if $t^* < \min\{x_i~|~x_i > r\}$, then we have $\Theta( \min\{x_i~|~x_i > r\}) - \Psi(\min\{x_i~|~x_i > r\}) < 0$, since $\frac{\partial}{\partial t^+}\Theta(t^+)  = -F_Y(t^+) \leq - F_X(t^+) = \frac{\partial}{\partial t^+}\Psi(t^+) = 0$ for all $t^+ <\min\{x_i~|~x_i > r\}$, which implies condition \eqref{eq:th_gains} does not hold. Second, if $t^* \in (\max\{x_i~|~x_i > r\}, ~b]$, given that $\frac{\partial}{\partial t^+}\Theta(t^+) - \frac{\partial}{\partial t^+}\Psi(t^+) = - F_Y(t^+) + F_X(t^+) = - F_Y(t^+) + 1$, together with $F_Y(t^+) \leq 1$, we obtain $\frac{\partial}{\partial t^+}\Theta(t^+) - \frac{\partial}{\partial t^+}\Psi(t^+) \geq 0$. Thus, $\Theta(t^+) - \Psi(t^+)$ is non-decreasing on $(\max\{x_i~|~x_i > r\}, ~b]$. Then, it follows that $\Theta(\max\{x_i~|~x_i > r\}) - \Psi(\max\{x_i~|~x_i > r\}) \leq \Theta(t^*) - \Psi(t^*) <0$. Thus, condition \eqref{eq:th_gains} does not hold. Finally, if $t^* \in [\min\{x_i~|~x_i > r\},~\max\{x_i~|~x_i > r\}]$, then there exist $x_a$, $x_b$ such that $t^* \in [x_a, x_b]$, $\Psi$ is linear and $\Theta$ is concave on this interval. Thus, $\Theta(t^*) - \Psi(t^*) < 0$ would imply that either $\Theta(x_a) - \Psi(x_a) < 0$ or $\Theta(x_b) - \Psi(x_b) < 0$. However, in either case, this implies then condition \eqref{eq:th_gains} does not hold.

In turn, we obtain by condition \eqref{eq:msd_losses} of Proposition \ref{prop:msd}, for all $t^- \in [a, r]$,
\begin{eqnarray*}
& &\int_a^{t^-} \big[F_Y(\tau) - F_X(\tau)\big]d\tau \geq 0 \\
&\Leftrightarrow &\int_{a}^{t^-} F_Y(\tau)d\tau - \int_{a}^{t^-} F_X(\tau)d\tau \geq 0 \\
&\Leftrightarrow & F^2_Y(t^-) - F^2_X(t^-) \geq 0 \\
&\Leftrightarrow & F^2_Y(t^-) \geq F^2_X(t^-),
\end{eqnarray*}
which gives the continuous case of condition \eqref{eq:th_losses} of Theorem \ref{th:discrete_MSD}. Similarly, we prove that under a discrete state-space condition \eqref{eq:th_losses} implies this condition. We prove this by contrapositive again, i.e., if the continuous case of condition \eqref{eq:th_losses} does not hold, then condition \eqref{eq:th_losses} does not hold either. Let $\Gamma(t^-) = F^2_Y(t^-) - F^2_X(t^-)$ for all $t^- \in \{y_i ~|~ y_i \leq r\}$. Assume now the continuous case of condition \eqref{eq:th_losses} does not hold for some $t^* \leq r$ in losses, then $\Gamma(t^*)< 0$. First, if $t^* \in [a, ~\min\{y_i~|~y_i \leq r\})$, evaluating $\frac{\partial}{\partial t^-}\Gamma(t^-) = F_Y(t^-) - F_X(t^-) = 0- F_X(t^-) \leq 0$ shows that $\Gamma(t^-)$ is non-increasing on $[a, ~\min\{y_i~|~y_i \leq r\})$. Hence, we obtain $\Gamma(\min\{y_i~|~y_i \leq r\}) \leq \Gamma(t^*) < 0$, which implies that condition \eqref{eq:th_losses} does not hold. Second, if $t^* > \max\{y_i~|~y_i \leq r\}$, since $\frac{\partial}{\partial t^-}\Gamma(t^-) = F_Y(t^-) - F_X(t^-) = 1 - F_X(t^-) \geq 0$ for all $t^- >\max\{y_i~|~y_i \leq r\}$, however, we have $\Gamma(\max\{y_i~|~y_i \leq r\}) \leq \Gamma(t^*) < 0$, i.e., a violation, which implies then condition \eqref{eq:th_losses} does not hold. Finally, assume that $t^* \in [\min\{y_i~|~y_i \leq r\},~\max\{y_i~|~y_i \leq r\}]$, then there exist $y_a$, $y_b$ such that $t^* \in [y_a, y_b]$, on which $F_Y^2$ is linear and $F_X^2$ is convex. Thus, $\Gamma(t^*) < 0$ would imply that either $\Gamma(y_a) < 0$ or $\Gamma(y_b) < 0$. Either case would then imply condition \eqref{eq:th_losses} does not hold. \qed
\endproof

\noindent \textbf{Proof of Theorem \ref{th:discrete_MWSD}} \label{appx:th_discrete_MWSD}

The proof of Theorem \ref{th:discrete_MSD} shows that conditions \eqref{eq:th_gains}
and \eqref{eq:th_losses} imply conditions \eqref{eq:msd_gains} and \eqref{eq:msd_losses} in Proposition \ref{prop:msd}, respectively. Here we only need to show that condition \eqref{eq:mwsd} of Proposition \ref{prop:mwsd} holds if and only if condition \eqref{eq:th_pwfs} of Theorem \ref{th:discrete_MWSD} holds.

We first prove the `if' part by contrapositive, i.e., if condition \eqref{eq:mwsd} of Proposition \ref{prop:mwsd} does not hold, then condition \eqref{eq:th_pwfs} of Theorem \ref{th:discrete_MWSD} does not hold either. Now assume that condition \eqref{eq:mwsd} of Proposition \ref{prop:mwsd} does not hold for some $t^*$ on $[t_d^-, t_d^+)$, where $t_d^- = \sup(\{a\} \cup \{t\leq r~|~F_X(t) \leq d^-,~F_Y(t) \leq d^- \})$ and $t_d^+ = \inf(\{b\} \cup \{t\geq r~|~F_X(t) \geq 1-d^+,~F_Y(t) \geq 1-d^+\})$, then $F_X(t^*) > F_Y(t^*)$, together with $\max\{F_X(t^*),~F_Y(t^*)\} > d^-$ and $\min\{F_X(t^*),~F_Y(t^*)\} < 1-d^+$, yields $F_X(t^*) > d^-$ and $F_Y(t^*) <1-d^+$. Moreover, choose $l \in \{1,\dots,n+1\}$ such that $t^* \in [y'_{l-1}, y_l')$, then $F_Y(t^*) = F_Y(y'_{l-1}) <1-d^+$. Together, these inequalities imply that $F_X(t^*) > \max\{F_Y(y'_{l-1}),~d^-\}$. However, evaluating the left-hand side of condition \eqref{eq:th_pwfs} for index $i=l$ gives
$$ \tilde{F}_X(y'_l) = \sum_{k|x_k<y'_l} p_k \geq \sum_{k|x_k \leq t^*} p_k = F_X(t^*) > \max\{F_Y(y'_{l-1}),~d^-\},$$
which implies that condition \eqref{eq:th_pwfs} does not hold.

Next, we prove the `only if' part by assuming that condition \eqref{eq:mwsd} of Proposition \ref{prop:mwsd} holds. In order to show condition \eqref{eq:th_pwfs} holds, we evaluate its left-hand side for an arbitrary $i \in \{1,\dots,n+1\}$ such that $F_Y(y'_{i-1}) < 1-d^+$ to obtain
$$ \tilde{F}_X(y'_i) = \sum_{k|x_k<y'_i} p_k = F_X(\max_k\{x_k~|~x_k<y'_i\}).$$
By condition \eqref{eq:mwsd}, it holds that $F_X(t) \leq F_Y(t) ~\forall~t \in [t_d^-, t_d^+)$, where $t_d^- = \sup(\{a\} \cup \{t\leq r~|~F_X(t) \leq d^-,~F_Y(t) \leq d^- \})$ and $t_d^+ = \inf(\{b\} \cup \{t\geq r~|~F_X(t) \geq 1-d^+,~F_Y(t) \geq 1-d^+\})$. First, suppose that $t_d^+ > r$. Then $t_d^+ = \min (\{b\} \cup \{t\geq r~|~F_X(t) \geq 1-d^+,~F_Y(t) \geq 1-d^+\})$. Now assume, by contradiction, that $t_d^+ < y'_i$. Since $F_Y(t)$ is non-decreasing, it follows then $F_Y(t_d^+) \leq F_Y(y'_{i-1}) < 1-d^+$, which contradicts the condition $F_Y(t_d^+) \geq 1-d^+$. Hence, $t_d^+ \geq y'_i$. Now suppose that $t_d^+ = r$. Then $r \geq y'_i$, and it follows that $\max_k \{x_k|x_k < y'_i\} < y'_i \leq r = t_d^+$. Together, these imply that $\max_k\{x_k~|~x_k<y'_i\} \in [t_d^-, t_d^+)$. Then, evaluating condition \eqref{eq:mwsd} of Proposition \ref{prop:mwsd} at $t=\max_k\{x_k~|~x_k<y'_i\}$ gives
$$F_X(\max_k\{x_k~|~x_k<y'_i\}) \leq F_Y(\max_k\{x_k~|~x_k<y'_i\}) \leq \sup_{t<y'_i} F_Y(t) \leq F_Y(y'_{i-1}) \leq \max\{F_Y(y'_{i-1}),~d^-\}.$$
Thus, condition \eqref{eq:th_pwfs} is satisfied. \qed
\endproof

\noindent \textbf{Proof of Lemma \ref{lem:milp_gains}} \label{appx:lem_milp_gains}

We prove the `if' part first. Assume that $(\varphi, ~\psi, ~\theta, ~z)$ is a feasible solution to constraints $\eqref{eq:milp_i} - \eqref{eq:milp_v}$, then $(\varphi, ~\psi, ~\theta, ~z)$ is also feasible to MILP problem $\eqref{eq:gains3} - \eqref{eq:gains7}$. Thus, the objective function \eqref{eq:gains3} value evaluated in $(\varphi, ~\psi, ~\theta, ~z)$ is greater than (or equal to) the value of $F^2_Y(x_i) + F^2_X(b) - F^2_X(x_i)$. Since $(\varphi, ~\psi, ~\theta, ~z)$ satisfies \eqref{eq:milp_v}, the objective function \eqref{eq:gains3} value is no greater than $F_Y^2(b)$. Together, these two results imply that for an arbitrary $x_i > r$, 
$$F^2_Y(x_i) + F^2_X(b) - F^2_X(x_i) \leq \sum_{k=1}^{n} p_k \varphi_{ik} + \sum_{k=1}^{n} p_k \psi_k - \sum_{k=1}^{n} p_k \theta_{ik} \leq F_Y^2(b), $$
where $\sum_{k=1}^{n} p_k \varphi_{ik} + \sum_{k=1}^{n} p_k \psi_k - \sum_{k=1}^{n} p_k \theta_{ik}$ provides an upper bound for $F^2_Y(x_i) + F^2_X(b) - F^2_X(x_i)$ for any feasible solution $(\varphi, ~\psi, ~\theta, ~z)$. Thus, condition \eqref{eq:th_gains} of Theorem \ref{th:discrete_MSD} holds.

We now prove the `only if' part by assuming condition \eqref{eq:th_gains} of Theorem \ref{th:discrete_MSD} holds. Take any $x_i > r$ and assume that $(\varphi^*, ~\psi^*, ~\theta^*, ~z^*)$ is the optimal solution to MILP problem $\eqref{eq:gains3} - \eqref{eq:gains7}$. Clearly, this optimal solution also satisfies constraints $\eqref{eq:milp_i} - \eqref{eq:milp_iv}$. Evaluating the objective function \eqref{eq:gains3} value yields 
$$ \sum_{k=1}^{n} p_k \varphi_{ik}^* + \sum_{k=1}^{n} p_k \psi_k^* - \sum_{k=1}^{n} p_k \theta_{ik}^* = F^2_Y(x_i) + F^2_X(b) - F^2_X(x_i) \leq F_Y^2(b), $$
which implies that \eqref{eq:milp_v} is satisfied. \qed
\endproof

\noindent \textbf{Proof of Lemma \ref{lem:milp_losses}} \label{appx:lem_milp_losses}

We first prove the `if' part by assuming that $\delta$ is a feasible solution to constraints $\eqref{eq:milp_vi} - \eqref{eq:milp_vii}$, which implies then $\delta$ is also a feasible solution to LP problem \eqref{eq:losses2}. Hence, the objective function value of LP problem \eqref{eq:losses2} evaluated at $\delta$ is no less than the value of $F_X^2(y_i)$. Given that $\delta$ satisfies constraint \eqref{eq:milp_vii}, then the objective function value of \eqref{eq:losses2} is less or equal to $F_Y^2(y_i)$. Therefore, combining these two results yields that for an arbitrary $y_i \leq r$,
$$ F_X^2(y_i) \leq \sum_{k=1}^{n} p_k \delta_{ik} \leq F_Y^2(y_i), $$
where $\sum_{k=1}^{n} p_k \delta_{ik}$ provides an upper bound for $F_X^2(y_i)$ for any feasible solution $\delta$, which thus implies that condition \eqref{eq:th_losses} of Theorem \ref{th:discrete_MSD} holds.

Next, we prove the `only if' part. Assume that condition \eqref{eq:th_losses} of Theorem \ref{th:discrete_MSD} holds. Then, take any $y_i \leq r$ and construct $\delta^*$ such that $\delta^*$ is optimal to LP problem \eqref{eq:losses2}. Clearly, constraint \eqref{eq:milp_vi} is satisfied and evaluating the objective function value of \eqref{eq:losses2} gives
$$ \sum_{k=1}^{n} p_k \delta_{ik}^* = F_X^2(y_i) \leq  F_Y^2(y_i). $$
Thus, \eqref{eq:milp_vii} is also satisfied. \qed
\endproof

\noindent \textbf{Proof of Theorem \ref{th:milp_msd}} \label{appx:th_milp_msd}

First, we prove $(i)$ by assuming that there exists $X \in \mathbb{X}$ such that $X \succeq Y$. It follows then by Theorem \ref{th:discrete_MSD} that conditions \eqref{eq:th_gains} and \eqref{eq:th_losses} hold. Condition \eqref{eq:th_gains} implies by Lemma \ref{lem:milp_gains} that there exist $\varphi \in \mathbb{R}^{n \times n}_+$, $\psi \in \mathbb{R}^{n}_+$, $\theta \in \mathbb{R}^{n \times n}_+$, and $z \in \{0, 1\}^{n \times n}$ that satisfy $\eqref{eq:milp_i} - \eqref{eq:milp_v}$. Recall that \eqref{eq:milp_v} holds for all $i \in N$ such that $x_i > r$ and is not required to hold in any state $i$ where $x_i < r$. Since each $x_i$ is an unknown decision variable, this requires introducing \eqref{eq:milp_viii} and augmenting \eqref{eq:milp_v} with a big $M$-term to obtain \eqref{eq:milp_ix}. This ensures that \eqref{eq:milp_ix} binds $\xi_i = 1$ if $x_i > r$, otherwise \eqref{eq:milp_ix} becomes redundant. Moreover, by Lemma \ref{lem:milp_losses}, condition \eqref{eq:th_losses} implies that there exists $\delta \in \mathbb{R}^{n^-\times n}_+$ that satisfies $\eqref{eq:milp_vi} - \eqref{eq:milp_vii}$. Therefore, $(i)$ of Theorem \ref{th:milp_msd} holds.

Next, we prove $(ii)$ by assuming that there exist $x = (x_1, \dots, x_n) \in \mathcal{X}$, $\varphi \in \mathbb{R}^{n \times n}_+$, $\psi \in \mathbb{R}^{n}_+$, $\theta \in \mathbb{R}^{n \times n}_+$, $z \in \{0, 1\}^{n \times n}$, $\xi \in \{0, 1\}^{n}$, and $\delta \in \mathbb{R}^{n^- \times n}_+$ that satisfy $\eqref{eq:milp_i} - \eqref{eq:milp_iv}$ and $\eqref{eq:milp_vi} - \eqref{eq:milp_ix}$. Using the same reasoning obtained in $(i)$, it can be shown that this is equivalent to that there exist $\varphi \in \mathbb{R}^{n \times n}_+$, $\psi \in \mathbb{R}^{n}_+$, $\theta \in \mathbb{R}^{n \times n}_+$, and $z \in \{0, 1\}^{n \times n}$ that satisfy $\eqref{eq:milp_i} - \eqref{eq:milp_v}$ and $\delta \in \mathbb{R}^{n^-\times n}_+$ that satisfies $\eqref{eq:milp_vi} - \eqref{eq:milp_vii}$. By Lemmas \ref{lem:milp_gains} and \ref{lem:milp_losses}, this equivalence implies that conditions \eqref{eq:th_gains} and \eqref{eq:th_losses} hold, respectively. Consequently, Theorem \ref{th:discrete_MSD} implies then $X \succeq Y$ and hence, $(ii)$ of Theorem \ref{th:milp_msd} holds. \qed
\endproof

\noindent \textbf{Proof of Lemma \ref{lem:milp_pwfs}} \label{appx:lem_milp_pwfs}

First, we prove the `if' part. Assume that $\zeta$ is a feasible solution to constraints $\eqref{eq:milp_fsd1} - \eqref{eq:milp_fsd2}$. Then, $\zeta$ must be feasible to MILP problem \eqref{eq:milp_fsd0}. Therefore, the objective function value of MILP problem \eqref{eq:milp_fsd0} evaluated at $\zeta$ is greater than (or equal to) the value of $\tilde{F}_X(y_i)$. Since $\zeta$ satisfies also \eqref{eq:milp_fsd2}, the objective function value of \eqref{eq:milp_fsd0} is less or equal to  $\max\big\{F_Y(y_{i-1}), ~d^-\big\}$ for all $i \in N$ such that $F_Y(y_{i-1}) < 1-d^+$. Consequently, for an arbitrary $i \in N$ such that $F_Y(y_{i-1}) < 1-d^+$, we obtain
$$ \tilde{F}_X(y_i) \leq \sum_{k=1}^n p_k\zeta_{ik} \leq \max\big\{F_Y(y_{i-1}), ~d^-\big\}. $$
For any feasible solution $\zeta$, $\sum_{k=1}^n p_k\zeta_{ik}$ provides an upper bound for $\tilde{F}_X(y_i)$. Since the choice of $y_i$ is arbitrary, this implies that condition \eqref{eq:th_pwfs} of Theorem \ref{th:discrete_MWSD} is satisfied.

We then prove the `only if' part by assuming that condition \eqref{eq:th_pwfs} of Theorem \ref{th:discrete_MWSD} holds. First, let $\zeta^*$ be the optimal solution to MILP problem \eqref{eq:milp_fsd0}. Then, $\zeta^*$ clearly satisfies constraint \eqref{eq:milp_fsd1}. Now we evaluate the objective function value of \eqref{eq:milp_fsd0} for any $i \in N$ such that $F_Y(y_{i-1}) < 1-d^+$ to obtain
$$ \sum_{k=1}^n p_k\zeta_{ik}^* = \tilde{F}_X(y_i) \leq \max\big\{F_Y(y_{i-1}), ~d^-\big\}, $$
which implies that \eqref{eq:milp_fsd2} is also satisfied. \qed
\endproof

\noindent \textbf{Proof of Theorem \ref{th:milp_mwsd}} \label{appx:th_milp_mwsd}

We first prove $(i)$ by assuming that there exists $X \in \mathbb{X}$ such that $X \succeq_{d^-}^{d^+} Y$. By Theorem \ref{th:discrete_MWSD}, it follows that $X \succeq Y$ and condition \eqref{eq:th_pwfs} holds. Then, according to $(i)$ in Theorem \ref{th:milp_msd}, $X \succeq Y$ implies that there exist $\varphi \in \mathbb{R}^{n \times n}_+$, $\psi \in \mathbb{R}^{n}_+$, $\theta \in \mathbb{R}^{n \times n}_+$, $z \in \{0, 1\}^{n \times n}$, $\xi \in \{0, 1\}^{n}$, and $\delta \in \mathbb{R}^{n^- \times n}_+$ that satisfy $\eqref{eq:milp_i} - \eqref{eq:milp_iv}$ and $\eqref{eq:milp_vi} -\eqref{eq:milp_ix}$. Moreover, by Lemma \ref{lem:milp_pwfs}, condition \eqref{eq:th_pwfs} implies that there exists $\zeta \in \{0, 1\}^{n \times n}$ that satisfies $\eqref{eq:milp_fsd1} - \eqref{eq:milp_fsd2}$. Together, these two results imply that $(i)$ of Theorem \ref{th:milp_mwsd} holds. 

We now prove $(ii)$ by assuming that there exist $x = (x_1, \dots, x_n) \in \mathcal{X}$, $\varphi \in \mathbb{R}^{n \times n}_+$, $\psi \in \mathbb{R}^{n}_+$, $\theta \in \mathbb{R}^{n \times n}_+$, $z \in \{0, 1\}^{n \times n}$, $\xi \in \{0, 1\}^{n}$, $\delta \in \mathbb{R}^{n^- \times n}_+$, and $\zeta \in \{0, 1\}^{n \times n}$ that satisfy $\eqref{eq:milp_i} - \eqref{eq:milp_iv}$, $\eqref{eq:milp_vi}-\eqref{eq:milp_ix}$, and $\eqref{eq:milp_fsd1} - \eqref{eq:milp_fsd2}$. By $(ii)$ in Theorem \ref{th:milp_msd}, there exists then $X \in \mathbb{X}$ such that $X \succeq Y$. In addition, Lemma \ref{lem:milp_pwfs} implies condition \eqref{eq:th_pwfs} holds. Collectively, Theorem \ref{th:discrete_MWSD} implies then $X \succeq_{d^-}^{d^+} Y$. Hence, $(ii)$ of Theorem \ref{th:milp_mwsd} holds. \qed
\endproof

\section{Heatmaps on asset compositions of MSD portfolios $X^{rf}$, $X^{md}$ and MWSD portfolios $X_{0.18}^{rf}$, $X_{0.18}^{md}$}\label{appx:heatmaps}

\begin{figure}[h!]
\includegraphics[width=\textwidth, keepaspectratio]{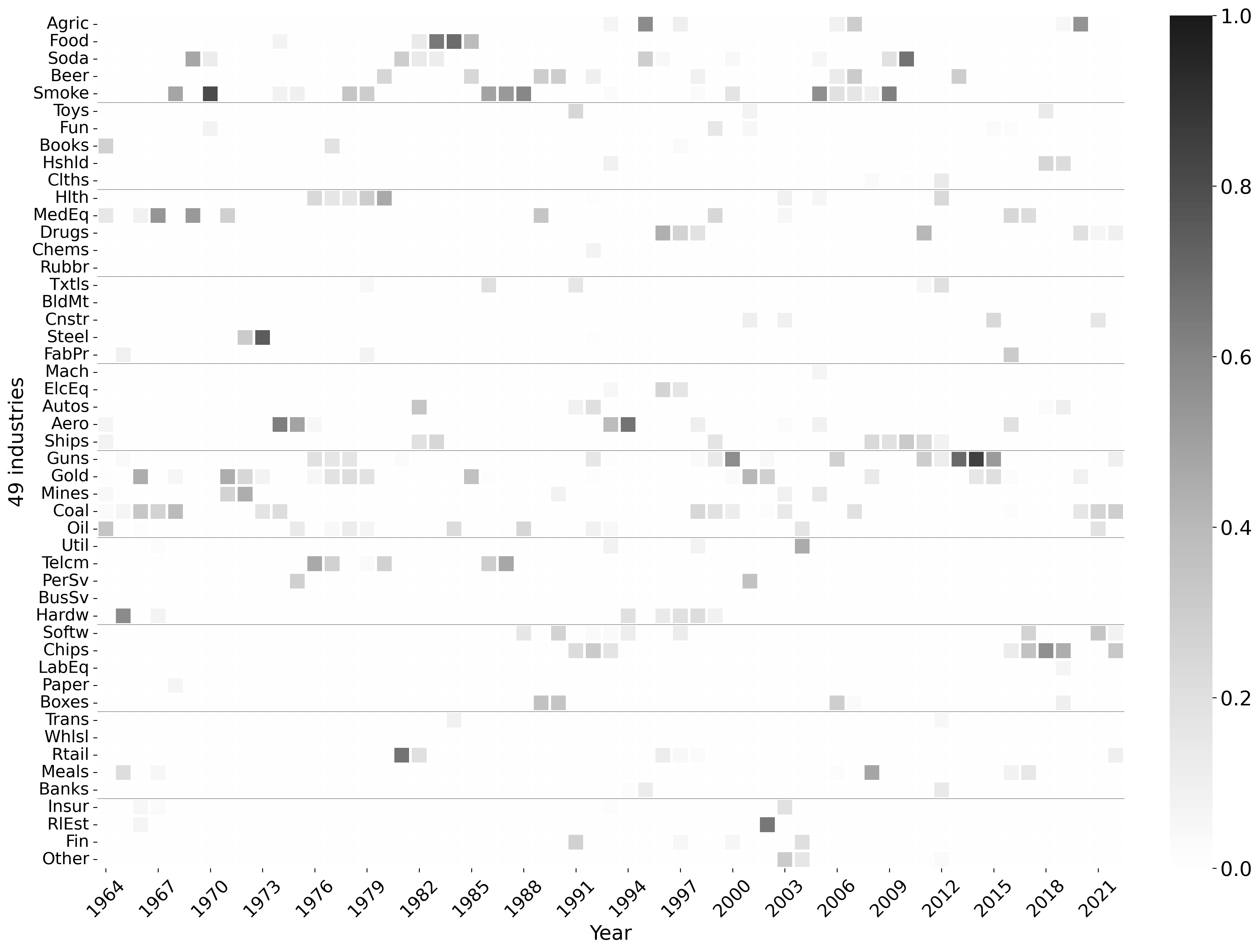}
\caption{Asset weights $\lambda_j \in [0,1]$ of the $m=49$ industries in MSD portfolio $X^{rf}$ for each year from 1964 to 2024.\label{ec_fig1}}
\end{figure}

\begin{figure}[h!]
\includegraphics[width=\textwidth, keepaspectratio]{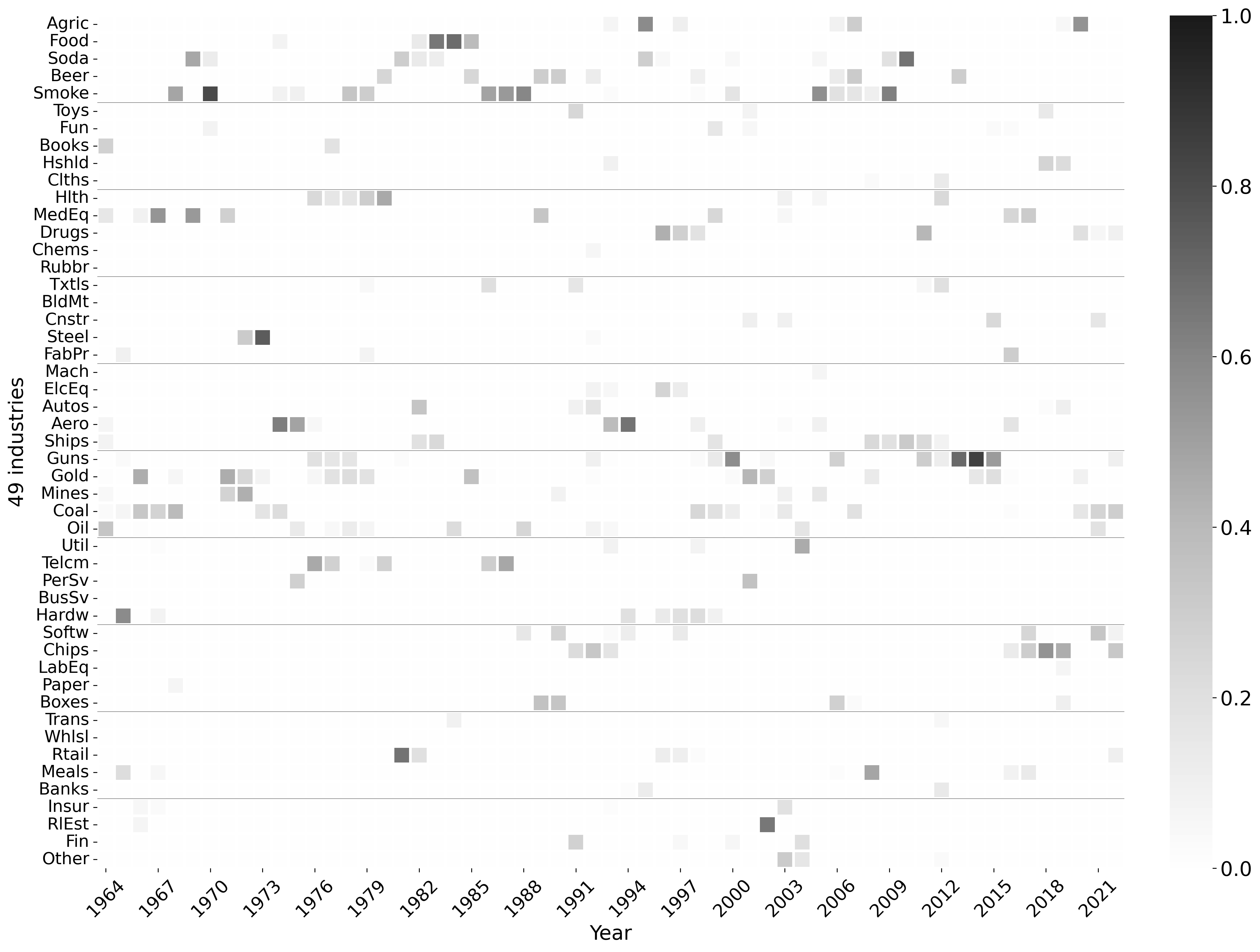}
\caption{Asset weights $\lambda_j \in [0,1]$ of the $m=49$ industries in MSD portfolio $X^{md}$ for each year from 1964 to 2024.\label{ec_fig2}}
\end{figure}

\begin{figure}[h!]
\includegraphics[width=\textwidth, keepaspectratio]{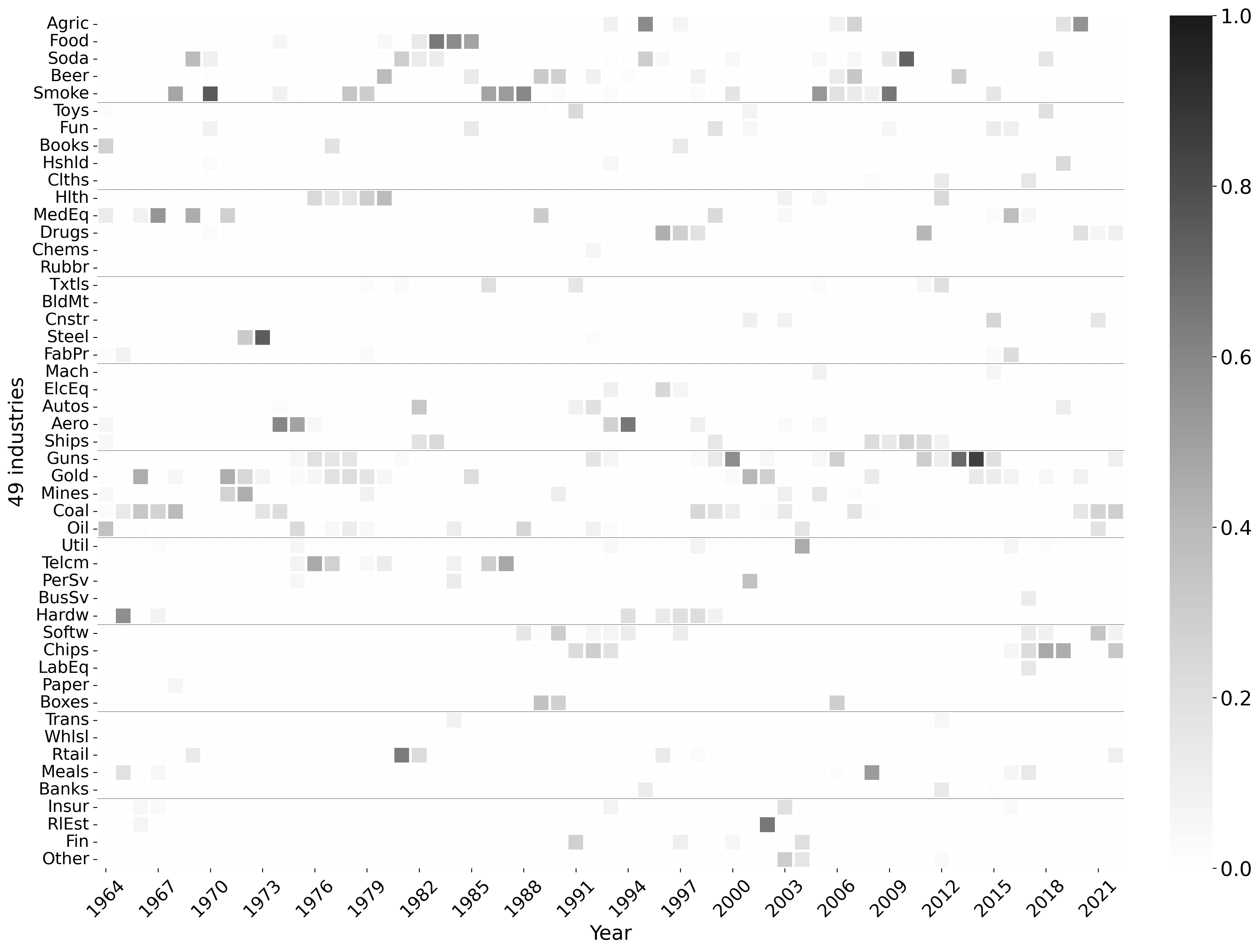}
\caption{Asset weights $\lambda_j \in [0,1]$ of the $m=49$ industries in MWSD portfolio $X_{0.18}^{rf}$ for each year from 1964 to 2024.\label{ec_fig3}}
\end{figure}

\begin{figure}[h!]
\includegraphics[width=\textwidth, keepaspectratio]{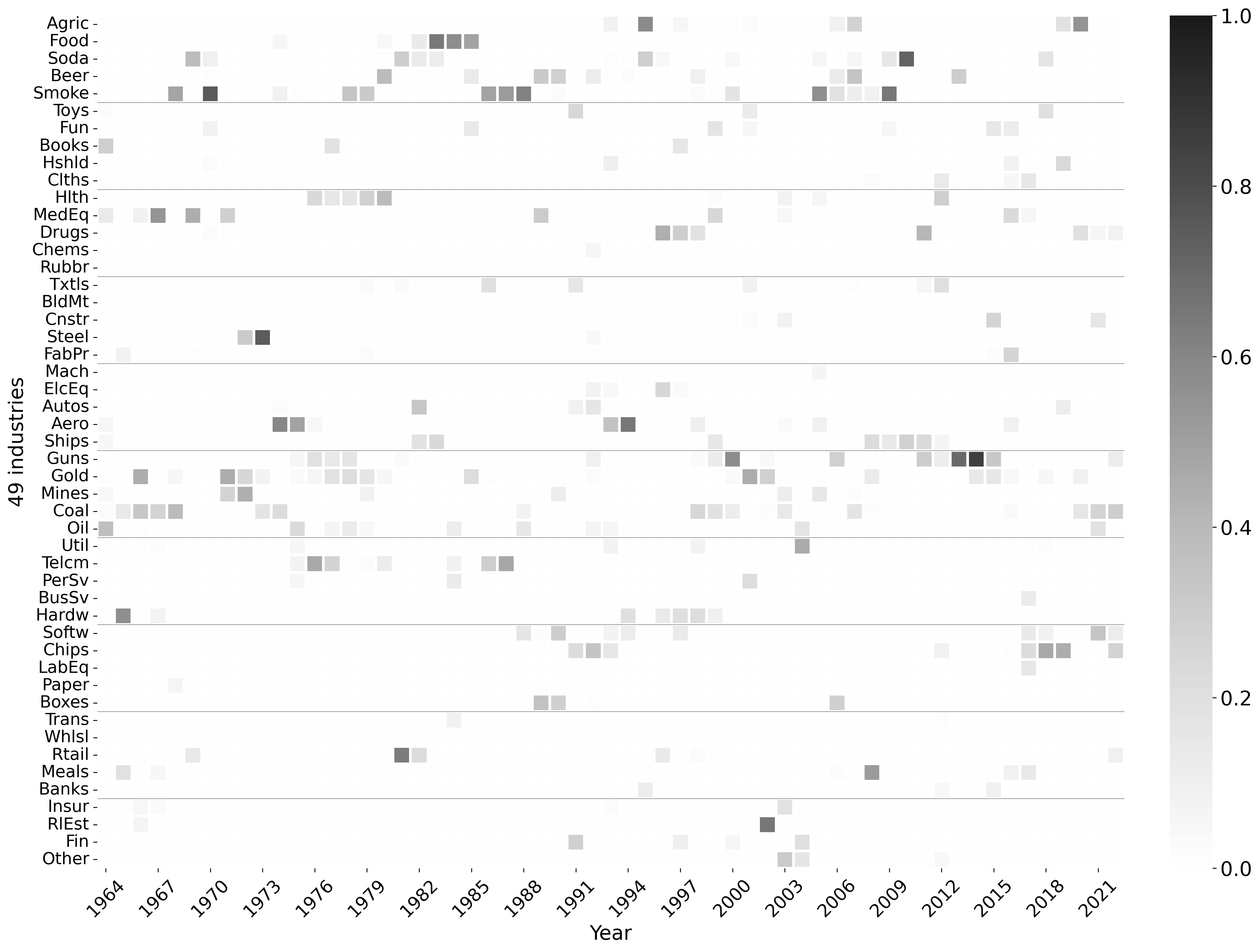}
\caption{Asset weights $\lambda_j \in [0,1]$ of the $m=49$ industries in MWSD portfolio $X_{0.18}^{md}$ for each year from 1964 to 2024.\label{ec_fig4}}
\end{figure}

\end{document}